\documentclass[journal]{IEEEtran}
\pdfoutput=1
\makeatletter
 \def\ps@headings{%
 \def\@oddhead{\mbox{}\scriptsize\rightmark \hfil \thepage}%
 \def\@evenhead{\scriptsize\thepage \hfil \leftmark\mbox{}}%
 \def\@oddfoot{}%
 \def\@evenfoot{}}
 \makeatother
\usepackage{amsthm,url}
\usepackage{amsmath}
\usepackage{amsfonts}
\usepackage{amssymb}
\usepackage{graphicx}
\usepackage{verbatim, color}

\newtheorem{thm}{Theorem}
\newtheorem{lemma}[thm]{Lemma}

\theoremstyle{definition}

\newtheorem{definition}{Definition}

\newcommand{\mybm}[1]{\mbox{\boldmath{$#1$}}}
\def\smallint{\begingroup\textstyle \int\endgroup}
\begin{document}

\title{Distributed Power Control and Coding-Modulation Adaptation in Wireless Networks using Annealed Gibbs Sampling}

%
%
%
%
%

\author{Shan~Zhou,
        Xinzhou~Wu,
        and~Lei~Ying
\thanks{}}


\maketitle
\begin{abstract}
In wireless networks, the transmission rate of a link is determined by received signal strength, interference from simultaneous transmissions, and available coding-modulation schemes. Rate allocation is a key problem in wireless network design, but a very challenging problem because: (i) wireless interference is global, i.e., a transmission interferes all other simultaneous transmissions, and (ii) the rate-power relation is non-convex and non-continuous, where the discontinuity is due to limited number of coding-modulation choices in practical systems. In this paper, we propose a distributed power control and coding-modulation adaptation algorithm using {\em annealed Gibbs sampling}, which achieves throughput optimality in an arbitrary network topology. We consider a realistic Signal-to-Interference-and-Noise-Ratio (SINR) based interference model, and assume continuous power space and finite rate options (coding-modulation choices). Our algorithm first decomposes network-wide interference to local interference by properly choosing a ``neighborhood'' for each transmitter and bounding the interference from non-neighbor nodes. The power update policy is then carefully designed to emulate a Gibbs sampler over a Markov chain with a continuous state space. We further exploit the technique of simulated annealing to speed up the convergence of the algorithm to the optimal power and coding-modulation configuration. Finally, simulation results demonstrate the superior performance of the proposed algorithm.
\end{abstract}




\section{Introduction}
Wireless communications have become one of the main means of
communications over the last two decades, in the form of both cellular (WWAN) and home/business access point (WLAN) communications \cite{80211g}. Recently, with the development of data centric mobile
devices, e.g., iPhone, we have seen a renewed interest in enabling
more flexible wireless networks, e.g., ad hoc networks and peer to
peer networks \cite{allerton10}. A key problem in the design of ad
hoc wireless networks is link-rate control, i.e., controlling
transmission rates of the links.  In wireless networks, the
transmission rate of a link is determined by received signal
strength, interference from simultaneous transmissions, and
available coding-modulation schemes. Because wireless interference
is global, and the rate-power relation is non-convex
\cite{EtkinPT07} and non-continuous, distributed link-rate control
in ad hoc wireless networks is a very challenging problem.

One approach to tackle the link-rate control problem in the
literature is to assume that the coding-modulation scheme is predetermined, i.e., all links use the same
coding-modulation scheme. In this case, there is an SINR threshold
associated with each link, and a transmission over the link can be
successfully decoded when the actual SINR is above the threshold.
This assumption is reasonable for voice-centric wireless networks
where the same voice codec is used at all devices. Under this
assumption, the link-rate control is translated into a power control
problem where the objective is to find a set of minimum transmission
powers such that the SINRs at all links are above the thresholds, if
possible at all. This problem has been well studied in the context
of the power control in cellular communications \cite{Yat_95} and
simple iterative algorithms can be shown to converge to the optimal
feasible power allocations.

Another approach in the literature is to assume the link rate is a
continuous function of the SINR of the link
\cite{Nee_03,TanChiSri_09}. For example, a model that has been
extensively adopted is to assume
$r_{ab}=\frac{1}{2}\log_2(1+SINR_{ab}),$ where $r_{ab}$ is the
transmission rate of link $(ab).$ In other words, it assumes that
for each SINR level, the capacity achieving coding-modulation is
available. Under this assumption,  the rate control problem again is
formulated as a power control problem where the objective is to find
a set of powers to maximize system utility defined upon achievable
rates $\sum_{ab}U_{ab}(r_{ab}),$ where $U_{ab}(\cdot)$ is the utility
function associated with link $ab.$ This problem is also well
understood in cellular networks given the recent advances in optimal
power control and rate assignment\cite{HandeRCW08,HuangBH06}, where
distributed iterative algorithms are shown to converge to the
utility maximizing power allocations, after introducing a small
signaling overhead to the cellular air interface. However, these
approaches ignore the non-convex nature of the problem and the
algorithms proposed here converge to the utility maximizing
operating point on \emph{Pareto boundary} of the rate region,
assuming all devices have to transmit all the time.
In the context of ad hoc networks, such approaches can be highly
sub-optimal since the time-sharing, or inter-link scheduling, nature
of the problem has to be considered due to highly non-convex
nature of the rate-power function. Towards this end, queue-length
distributed scheduling is shown to be throughput optimal
\cite{JiaWal_08,LiuYiPro_08,RajShaShi_09,NiTanSri_10}, through the use of MCMC (Markov Chain Monte Carlo) models.
These results however assume collision-based interference model, which in general is over-conservative, and assume fixed transmit power and coding-modulation scheme. Both transmit power and coding-modulation can be adaptively chosen in practical systems. For example, in 802.11g, eight rate options are available, and many 802.11 chip solutions have capability of packet to packet power control with very good granularity (0.5dBm).
{\color{black} Gibbs sampling based distributed power control algorithms have also been developed in \cite{QiaZhaChi10,YanSagZha11}. However, these work again assumes the rate is a continuous function of the SINR level, and ignores the fact that the set of coding-modulation schemes is finite in practical systems.}

In this paper, we extend the framework in \cite{JiaWal_08,LiuYiPro_08,RajShaShi_09,NiTanSri_10}
to develop a distributed joint power control and rate scheduling
algorithm for wireless networks based on the SINR-based interference
model. We assume that each node has a finite number of
coding-modulation choices, but can continuously control transmit power. We propose a distributed algorithm that maximizes the sum of
weighted link rates $\sum_{(ab)} w_{ab} r_{ab},$ where $r_{ab}$ is
the rate of link $ab.$ The main results of this paper are
summarized below:
\begin{itemize}
\item  We consider realistic SINR-based interference model, where a transmission interferes with all other simultaneous transmissions in the network. Our algorithm decomposes network-wide interference to local interference by properly choosing a ``neighborhood'' for each node, and bounding the interference from non-neighbor nodes.

\item We assume continuous power space and finite coding-modulation choices (rate options). The objective of the algorithm is to find a power and coding-modulation configuration that maximizes the sum of weighted link-rates
    \begin{equation} \label{eq: obj_1}
  \begin{array}{r c l}
    \max_{\mathbf{p}, \mathbf{m}} &  & \sum_{(ab) \in {\cal E}}{q_{ab}(t)r_{ab}(m_{ab}, \mathbf p)} \\
    \hbox{subject to} &  & \sum_{b:(ab)\in {\cal E}} p_{ab} \leq p^{\rm max}_a, \forall a\in {\cal V},
  \end{array}
\end{equation}
where $q_{ab}(t)$ is the queue length of link $(ab)$ at time slot $t,$
\footnote{We use $q_{ab}(t)$ as the link weight so that the algorithm is throughput optimal when problem is solved at each time slot.}
$\mathbf p$ is a vector containing the power levels of all the links in the network, $p^{\rm max}$ is the maximum power constraint, and $m_{ab}$ is the coding-modulation scheme.  Due to the nonconvexity and discontinuity of $r_{ab}(\cdot),$  optimization problem (\ref{eq: obj_1}) is very hard to solve in general. Motivated by recent breakthrough of using MCMC to solve MaxWeight scheduling in a completely distributed fashion, we propose a power and coding-modulation update algorithm that emulates a Gibbs sampler over a Markov chain with a continuous state space (the power level of a transmitter is assumed to be continuously adjustable).

\item The algorithm based on the Gibbs sampling may be trapped in a local-optimal configuration for an extended period of time. To overcome this problem, we exploit the technique of simulated annealing to speed up the convergence to the optimal power and coding-modulation configuration. The convergence of the algorithm under annealed Gibbs sampling is proved. From the best of our knowledge, this is the first algorithm that uses annealed Gibbs sampling in a distributed fashion with continuous sample space and has provable convergence.
\end{itemize}

\section{System Model}
We consider a wireless network with single-hop traffic flows. The network is modeled as a graph ${\cal G}=({\cal V},{\cal E}),$ where ${\cal V}$ is the set of nodes, and ${\cal E}$ is the set of directed links. Let $n=|{\cal E}|$ denote the number of links. We assume that time is slotted. Each transmitter $a$ maintains a buffer for each outgoing link $(ab),$ if there is a flow over link $(ab).$ Note that even if there are multiple flows over link $(ab),$ a single queue is sufficient for maintaining the stability of the network. The queue length in time slot $t$ is denoted by $q_{ab}(t).$ Each transmitter $a$ has limited total transmit power $p^{\rm max}_a,$ and $p_{ab}(t)$ denotes the transmit power of link $(ab)$ at time slot $t.$

We assume all links have stationary channels, and each transmitter $a$ can tune its transmit power continuously from $0$ to $p^{\rm max}_a,$ but the number of feasible coding-modulation choices is finite. Each coding-modulation associates with a fixed data rate, and a minimum SINR requirement. Thus, the data rate of a link is a step function of the SINR of the link.
The SINR of link $(ab)$ is
\begin{equation}\label{eq: sinr}
  \gamma_{ab}(t) = \frac{p_{ab}(t) g_{ab}}{n_b + \sum_{(xy)\in {\cal E}, (xy) \not= (ab)}p_{xy}(t) g_{xb}},
\end{equation}
where $n_b$ is the variance of Gaussian background noise experienced by node $b,$ and $g_{ab}$ is the channel gain from node $a$ to node $b.$ In this paper, all $n_b$s and $g_{ab}$s are assumed to be fixed, i.e. we consider stationary channels.

Denote by $r_{ab}(t)$ the transmission rate of $(ab)$ at time slot $t,$ and $A_{ab}(t)$ the number of bits that arrive at the buffer of the transmitter of link $(ab)$ at the end of time slot $t.$ Then, the queue length $q_{ab}(t)$ evolves as following:
\begin{equation}\label{eq: QueueingDynamic}
    q_{ab}(t+1) = \left[q_{ab}(t) - r_{ab}(t)\right]^+ + A_{ab}(t),
\end{equation}
where $[x]^+ = \max\{0,x\}.$

Let ${\cal P}\subset \mathbb{R}^n$ denote the set of all feasible power configurations of the network, i.e., $${\cal P} = \{{\mathbf p}: \sum_{b: (ab)\in {\cal E}} p_{ab}\leq p^{\rm max}_a, p_{ab}\geq 0\}.$$ For each link $(ab),$ given a power configuration $\mathbf p,$ the SINR of the link $\gamma_{ab}$ is determined by equality (\ref{eq: sinr}). The transmission rate can be written as $r_{ab}(m_{ab}, \mathbf{p}),$ where $m_{ab}$ is the coding-modulation scheme.

In this paper, we assume a transmitter always selects the coding-modulation scheme with the highest rate under the given SINR. {\color{black}Each coding-modulation scheme has a minimum requirement on the SINR level.} So $m_{ab}$ is a function of $\mathbf p,$ and rate $r_{ab}$ can be written as a function of $\mathbf{p}:$ $r_{ab} = r_{ab}(\mathbf p).$
Then, we define ${\mathcal R}$ as the set of achievable rate vectors under feasible power configurations and modulations, i.e.,
$${\cal R} = \{{\mathbf r}({\mathbf p}): {\mathbf p} \in {\cal P}\}.$$

The \emph{capacity region} of the network is the set of all arrival rate vectors $\mybm{\lambda}$ for which there exists a power control algorithm that can stabilize the network, i.e., keep the queue lengths from growing unboundedly. It is well known that the capacity region is \cite{TasEph_92}:
\begin{equation}\label{eq: DefCpctRg}
    \Lambda = \{\mybm{\lambda} | \exists \mybm{\mu} \in Co(\mathcal R), \mybm {\lambda} \prec \mybm{ \mu}\},
\end{equation}
where $Co(\mathcal R)$ is the convex hull of the set of achievable rates with feasible power configurations, and $\prec$ denotes componentwise inequality. A power control and coding-modulation adaptation algorithm is said to be \emph{throughput optimal} if it can stabilize the network for all arrival rates in the capacity region $\Lambda.$


It is well-known that if a rate control algorithm can solve the MaxWeight problem \cite{TasEph_92} for each time slot, then the algorithm is throughput optimal. The focus of this paper is to develop a power-control and coding-modulation adaptation algorithm to solve the following MaxWeight problem:
\begin{equation} \label{eq: obj}
  \begin{array}{r c l}
    \max &  & \sum_{(ab) \in {\cal E}}{r_{ab}(\mathbf p)q_{ab}(t)} -\epsilon \sum_{(ab) \in {\cal E}} p_{ab} \\
    \hbox{subject to} &  & {\mathbf p}\in {\mathcal P}.
  \end{array}
\end{equation}
Recall that since $r_{ab}(\mathbf{p})$ is a step function of $p_{ab},$ multiple power configurations may result in the maximum weighted sum. {\color{black} We therefore added a penalty function $-\epsilon \sum_{(ab)}p_{ab}$ with a small $\epsilon$ in the objective function so that the algorithm yields a power configuration whose weighted sum-rate is close to the optimal one and its sum power is small. Without this penalty term, the algorithm may result in a solution with maximum sum weighted rate but large $\sum_{(ab)} p_{ab}.$ This penalty term makes sure the proposed algorithm is energy efficient.} 

\section{Algorithm}


We are interested in obtaining the optimal power, coding, and modulation configuration that maximizes the weighted-sum-rate while minimizing the total transmit power.  We can solve this problem by constructing a Markov chain whose state is the power configuration, and the stationary density satisfies
\begin{equation}\label{eq: Gibbsdensity}
    \pi(\mathbf{p}) = \frac{1}{Z(K)}e^{\frac{1}{K}(\sum_{(ab)\in {\cal E}} r_{ab}({\mathbf p})q_{ab}(t)- \epsilon \sum_{(ab)\in \mathcal E} p_{ab})}
\end{equation}
Then, letting $K \rightarrow 0,$ the Markov chain is in state $\mathbf p^*,$ the optimal solution to problem (\ref{eq: obj}) {\color{black} with probability $1-\delta$ for any $\delta>0.$}. See \cite{Kau_07} \cite{Bre_99} for detail.

Gibbs sampler is a classical way to construct such a Markov Chain with stationary distribution (\ref{eq: Gibbsdensity}). Given the current state ${\mathbf p}(t)=\mathbf{p},$ the Gibbs sampler selects a link, say $(ab),$ in a predetermined order and changes the transmit power to $p_{ab}$ with probability
\begin{align*}
   \pi(p_{ab}| \mathbf p_{-ab})  = \frac{\pi(p_{ab}, \mathbf p_{-ab})}{\int_{0}^{p_a^{\rm max}}\pi(\hat{p}_{ab}, \mathbf p_{-ab}) d\hat{p}_{ab}} \nonumber
\end{align*}  where $\mathbf{p}_{-ab}$ denotes the vector of transmit powers except that of link $(ab).$ It can be verified that the stationary distribution of this Markov chain is (\ref{eq: Gibbsdensity}) by detailed balance equation, i.e., $\pi(p_{ab}| \mathbf p_{-ab}))  \pi( (p^\prime_{ab},\mathbf p_{-ab}) ) = \pi(p^\prime_{ab}| \mathbf p_{-ab}) \pi( (p_{ab},\mathbf p_{-ab}) ).$ Therefore, if the power vector is updated according to this Markov chain, it will converge to $\mathbf{p}^*$ with probability one when $K\rightarrow 0.$

There are, however, several difficulties in using Gibbs sampler for distributed rate control in wireless networks.
\begin{enumerate}
\item First, to compute the conditional distribution, link $(ab)$ must know the rate of all the links in the network, which incurs significant communication overhead.

\item Second, the updating sequence of a Gibbs sampler is predefined, which results in the need of a central controller.

\item Further, when $K$ is close to zero,
\begin{align*}
   &\pi(p_{ab}| \mathbf p_{-ab})= \frac{\pi(p_{ab}, \mathbf p_{-ab})}{\int_{0}^{p_a^{\rm max}}\pi(\hat{p}_{ab}, \mathbf p_{-ab}) d\hat{p}_{ab}} \nonumber\\
  =& \frac{e^{(\sum_{(xy)\in \mathcal E} r_{xy}(p_{ab}, \mathbf p_{-ab})q_{xy} - \epsilon \sum_{(ab)\in \mathcal E} p_{ab})/K}}{ \int_{0}^{p_a^{\rm max}} e^{(\sum_{(xy)\in \mathcal E} r_{xy}(\hat{p}_{ab}, \mathbf p_{-ab})q_{xy} - \epsilon \sum_{(ab)\in \mathcal E} p_{ab})/K}d\hat{p}_{ab}}\\
  &\rightarrow 1,
\end{align*} for $p_{ab}$ such that $$p_{ab}\in \arg\max_{\hat{p}_{ab}} \sum_{(xy)\in \mathcal E} r_{xy}(\hat{p}_{ab}, \mathbf p_{-ab})q_{xy} - \epsilon \sum_{(ab)\in \mathcal E} p_{ab}.$$ In other words, the power configuration may stay in a local optimum for a long period of time. This is in fact a critical weakness of MCMC methods. We will use simulated annealing technique in our algorithm to overcome this weakness.
\end{enumerate}

%

\subsection{Neighborhood and Virtual Rate}

To overcome the global interference, we note that because of channel attenuation, interference caused by a remote transmitter in general is negligible. We therefore define a neighborhood for each node $a.$ We say a node $b$ is a \emph{one-hop neighbor} of node $a$ if $\max\{g_{ab}, g_{ba}\}\geq \alpha,$ i.e., the channel gain is above certain threshold. Denote by ${\cal N}^1(a)$  the set of one-hop neighbors of node $a,$ where the superscript indicates it is the set of one-hop neighbors. We further denote by ${\cal N}^2(a)$ the set of two-hop neighbors of node $a,$ i.e., node $b$ belongs to ${\cal N}^2(a)$ if $b
\not \in {\cal N}^1(a),$ and $b\not =a,$ and there exists a node $c$ such that $c\in{\cal N}^1(a)$ and $b\in{\cal N}^1(c).$



Then, we define $\Upsilon_{ab}(t)$ to be
\begin{align*}
&\displaystyle \Upsilon_{ab}(t)=\sum_{(xy):x\in{\cal N}^1(b)} p_{xy}(t) g_{xb} + \hat{n}_b
\end{align*} which is called \emph{noise+partial-interference} at node $b$ for link $(ab)$ at time slot $t.$ {\color{black} We further define $\hat{\Upsilon}_{ab}(p_{xy})$ to be the noise+partial-interference at node $b$ if link $(xy)$ uses transmit power $p_{xy},$ and $\Upsilon_{ab}^m$ to be the maximum noise+partial-interference allowed to achieve the SINR requirement of coding-modulation scheme $m,$ while link $(ab)$ does not change its power level.} Let $\hat{n}_b$ denote $n_b + \xi_{ab},$ where $\xi_{ab}$ is an upper bound on the interference experienced at node $b$ from the non-neighboring transmitters of $b.$ We assume $\xi_{ab}$ is known. By including this upper bound $\xi_{ab}$ in the SINR computation, we guarantee that the SINR of link $(ab)$ is a function of its neighbors' transmit powers and is independent of non-neighbor nodes. This localizes the interference.

Given the definition of noise+partial-interference, the virtual rate of link $(ab)$ is defined to be
\begin{equation}\label{eq: VirtualRate}
    \tilde{r}_{ab}(t) = r_{ab}\left(m_{ab},\frac{{p}_{ab}g_{ab}}{\Upsilon_{ab}(t)}\right).
\end{equation}

Observe that although the power level is continuous, and the SINR of neighboring links are continuous, the virtual rate choices are discrete and finite. Suppose $(ab)$ is changing its power and link $(xy), y\in \mathcal N^1(a),$ is affected. For each coding-modulation scheme of link $(xy),$ $m\in \mathcal M_{xy},$ let $\gamma_{xy}^m$ denote the SINR requirement of $m,$ and $\Upsilon_{xy}^m$ denote the corresponding noise+partial-interference requirement. Assuming the transmit powers of all other nodes are fixed, the power level $p_{ab}\in [0, p_a^{\rm max}],$ such that $\hat{\Upsilon}_{xy}(p_{ab}) = {\Upsilon}_{xy}^m,$ is called a \emph{critical power}, which is highest power node $a$ can use for coding-modulation choice $m$ to be feasible over link $(xy).$


\subsection{Decision Set}

To overcome the issue of predefined update sequence in classic Gibbs sampling, we adopt the technique proposed in \cite{NiTanSri_10} to generate a \emph{decision set} at the beginning of each time slot.
\begin{definition}
A decision set $\mathcal D$ is a set of transmitters such that, for any two transmitters $a$ and $x$ in $\mathcal D,$
$x \not \in {\cal N}^1(a) \cup {\cal N}^2(a).$
\end{definition}
Clearly, two transmitters $a$ and $x$ in the decision set are not one-hop or two-hop neighbors.  In the proposed algorithm, only the links in the decision set are allowed to update their transmit powers. By properly generating a decision set, the evolution of power configuration is a reversible Markov chain with stationary density (\ref{eq: Gibbsdensity}).

\subsection{Required Information}
We further assume that node $a$ has the following knowledge:
\begin{itemize}
\item the channel gain from node $a$ to its one-hop neighbors, i.e., $g_{ay}$ for all $y \in {\cal N}^1(a).$
\item for each link whose receiver is $a'$s one-hop neighbor, i.e., $(xy): y\in \mathcal N^1(a):$
\begin{itemize}
\item the virtual transmit power of $(xy),$ i.e., $\tilde{p}_{xy}(t-1).$ {\color{black} where the virtual transmit power is the intermediate power level obtained during each time slot of the proposed algorithm. The real transmit power is updated at the end of every super time slot. So the virtual power is an intermediate result obtained and maintained during the calculation and is not the actual transmit power. }
\item the channel gain of $(xy),$ i.e., $g_{xy}.$
\item the virtual partial-interference-plus-noise of $(xy),$ i.e., $\tilde{\Upsilon}_{xy}(t).$\footnote{calculated based on virtual transmit powers}
\item the queue length of $(xy),$ i.e., $q_{xy}.$
\item the feasible modulations of $(xy),$ and the maximum allowed noise+partial-interference, ${\Upsilon}_{xy}^m,$ of each modulation $m.$
\end{itemize}
\end{itemize}
Notice that $\tilde{p}_{xy}(t-1), q_{xy},$ and $\tilde{\Upsilon}_{xy}(t)$ change over time. We will explain the way node $a$ obtains these values from node $x$ in the algorithm. Further we assume all channel gains are known to $a,$ and are fixed. The feasible modulations and the minimum SINR requirement for each modulation are assumed to be known a-priori, and do not need to be exchanged.

\subsection{Distributed Power Control and Coding Modulation Adaptation Algorithm}

We now present the proposed algorithm, where the evolution of power configuration emulates a Gibbs sampler. To improve the convergence of the Gibbs sampler, we exploit the technique of simulated annealing \cite{GemGem_90}.

We group every $T$ time slots into a super time slot. In each super time slot, we run the algorithm for $T$ times in the background of each node. In ${t}^{\rm th}$ time slot of a super time slot, the value $K$ is set to be $K_{t}=\frac{K_0}{\log(2+{t})},$ where $K_{t}$ is the ``temperature'' in the terminology of simulated annealing, and $K_0$ is a positive constant that can be tuned to control the convergence of the proposed algorithm. The idea of the simulated annealing is to start with a high temperature (large $K$) under which the Markov chain mixes rapidly. Then by slowly decreasing the temperature, the state of the Markov chain will converge to the optimal configuration. It has been well-known that annealing can significantly reduce the convergence time. The structures of time slot and super time slot are illustrated in Figure \ref{fig: timeslot} and \ref{fig: supertimeslot}.

\begin{figure}[hbt!]
\begin{center}
  \includegraphics[width=2.8in]{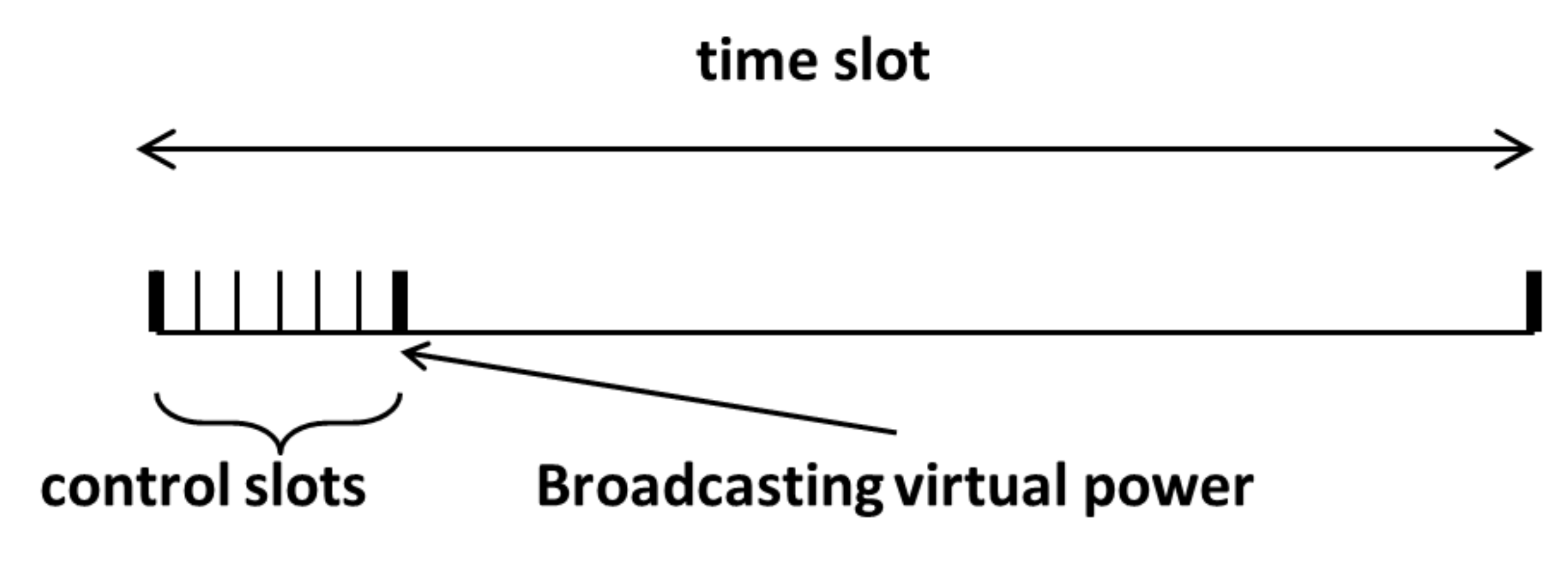}\\
  \end{center}\caption{Time slot structure}\label{fig: timeslot}
\end{figure}

\begin{figure}[hbt!]
\begin{center}
  \includegraphics[width=3.8in]{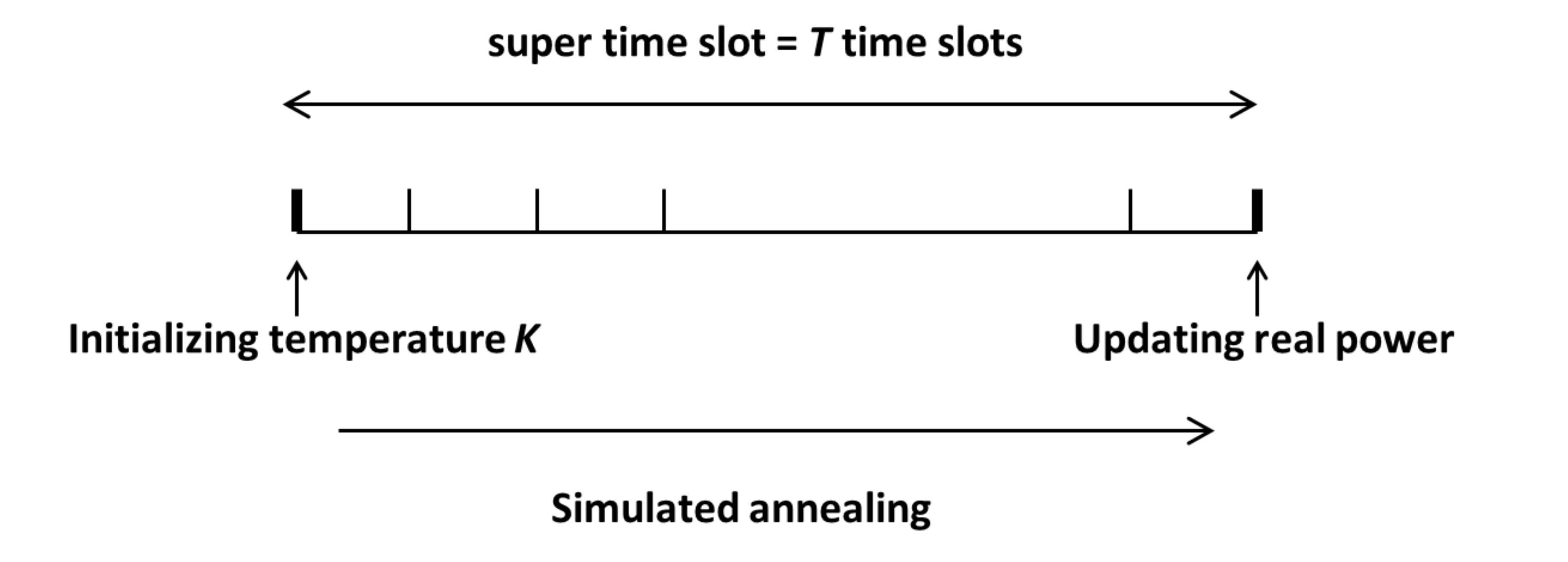}\\
  \end{center}\caption{Super time slot structure}\label{fig: supertimeslot}
\end{figure}

All nodes maintain virtual power $\tilde{p},$ and the initial power configuration $\tilde{p}_{ab}(0)=0,$ known by all the nodes in the network. The following algorithm describes the process of updating virtual power configuration following an annealed Gibbs sampler. The real transmit power is then determined based on actual SINR.

At $t^{th}$ time slot of a super time slot, the algorithm works as follows:
\begin{enumerate}

\item[(1)]{\bf Generating decision set:} Each time slot consists of $W$ control slots at the beginning. A decision set is determined at the end of the $W$ control slots. Only the transmitters in the decision set update their virtual power levels at this time slot. At time slot $t,$ transmitter $a$ contends for being in the decision set as follows:
\begin{enumerate}
\item[(i)] Node $a$ uniformly and randomly selects an integer backoff time $T_a$ from $[0, W-1]$ and wait for $T_a$ control slots.
\item[(ii)] If $a$ receives an INTENT message from another transmitter $x$ such that $$x\in {\cal N}^1(a)\cup {\cal N}^2(a)$$ before control slot $T_a+1,$ node $a$ will not be included in the decision set in this time slot. Here, we assume the INTENT message from $x$ has the id of $x$ and the signal is strong enough, so that $x$'s one-hop and two-hop neighbors, e.g., node $a,$ know this INTENT message is sent by $x.$
\item[(iii)] If node $a$ senses a collision of INTENT messages from nodes $x,$ $x\in {\cal N}^1(a)\cup {\cal N}^2(a),$ $a$ will not be in the decision set in this time slot.
\item[(iv)] If node $a$ does not receive any INTENT message from its one-hop or two-hop neighbors before control slot $T_a+1$, node $a$ will broadcast an INTENT message to its one-hop and two-hop neighbors in control slot $T_a+1$.
    \begin{enumerate}
      \item If the INTENT message from node $a$ collides with another INTENT message sent by node $x\in {\cal N}^1(a)\cup {\cal N}^2(a),$ $a$ will not be in the decision set in this time slot.
      \item If there is no collision, node $a$ will be included in the decision set in this time slot.
    \end{enumerate}
\end{enumerate}
We note that $W$ is selected to be large enough so that the collision of the INTENT messages happens with low probability.


\item[(2)] {\bf Link selection:} Let $d_a$ denote the outgoing degree of node $a,$ i.e., $d_a = |\{b: (ab)\in {\mathcal E}\}|.$
In this step, each transmitter $a\in {\cal D}$ selects an outgoing link $(ab)$ to update its virtual power $\tilde{p}_{ab}$ as following:
\begin{itemize}
  \item [(i)]If there was an active outgoing link $(ab)$ such that $\tilde{p}_{ab}(t-1)>0,$ $a$ will update the power of link $(ab)$ in time slot $t,$ with probability $\frac{1}{d_a}.$
  \item [(ii)] If there was no active outgoing link $(ab)$ such that $\tilde{p}_{ab}(t-1)>0,$ $a$ uniformly randomly selects a link $(ab)$ from its $d_a$ outgoing links, and then updates its virtual power $\tilde{p}_{ab}.$
\end{itemize}

\item[(3)]{\bf Critical power level computation:} Node $a$ computes the critical power level $\tilde{p}_{c,m}(ab, xy)$ as follows:
\begin{enumerate}
  \item[(i)] Node $a$ computes the critical partial-noise-plus-interference of link $(xy)$ corresponding to each $\gamma_{xy}^m, m\in {\mathcal M}_{xy}:$
  $$\tilde{\Upsilon}_{xy}^m = \frac{\tilde{p}_{xy}(t-1)g_{xy}}{\gamma_{xy}^m}.$$
  \item[(ii)] Node $a$ computes the critical power of $\tilde{p}_{c,m}(ab,xy),$ such that when link $(ab)$ uses this power level, the resulting partial-noise-plus-interference of link $(xy)$ is $\tilde{\Upsilon}_{xy}^m:$
  \begin{align*}
  &\tilde{p}_{c,m}(ab, xy)  \\
  =&\min\left\{p^{\rm max}, \left[ \tilde{p}_{ab}(t-1) + \frac{\tilde{\Upsilon}_{xy, m} - \tilde{\Upsilon}_{xy}(t-1)}{g_{ay}}\right]^+\right\}.
\end{align*}
  (Some modulations of link $(xy)$ need very high SINR, which cannot be achieved even link $(ab)$ reduces its power to 0. For these modulations, we just let $\tilde{p}_{c,m}(ab, xy)$ be zero, we will consider the $0$ critical power separately in the following step.)
\end{enumerate}

\item[(4)] {\bf Virtual rates computation:} Now for each node $a$ in the decision set ${\cal D}$, it computes the virtual rate of link $(ab),$ and the virtual rate of the links whose receiver is $a'$s neighbor as following:
\begin{itemize}
    \item[(i)] Arrange the critical power levels $$\{\tilde{p}_{c,m}(ab, xy), \forall (xy) \quad {\textrm s.t. }\quad y\in {\cal N}^1(a)\}$$ in ascending order, denoted by
     $$0=\tilde{p}_{c,0}<\tilde{p}_{c,1}<\cdots <\cdots = p_a^{\rm max}$$
    \item[(ii)] Compute the SINR of link $(xy),$ when the power of link $(ab)$ is zero:
    $$\gamma_{xy}^0 = \frac{\tilde{p}_{xy}(t-1) g_{xy}}{\tilde{\Upsilon}_{xy}(t-1) - \tilde{p}_{ab}(t-1) g_{ay}}.$$
    Further, find the coding-modulation of link $(xy)$ with the largest transmission rate corresponding to this SINR. Let the coding-modulations of all the neighboring links of link $(ab)$ be denoted by a vector ${\mathbf m}^0$.
    \item[(iii)] Given this initial coding-modulation vector ${\mathbf m}^0,$ $a$ obtains the coding-modulation vector:
    $${\mathbf m}^i,$$ corresponding to each critical power $\tilde{p}_{c,i}, i=1, 2, \cdots.$
    \item[(iv)] Obtain the rates $\tilde{r}_{(xy)}(\tilde{p}_{c,i})$ related to the coding-modulations $m_{(xy)}^i$ of neighboring links, when $$\tilde{p}_{ab} \in [\tilde{p}_{c,i}, \tilde{p}_{c, i+1}), \, i=0, 1, \cdots.$$  Note that for each link, the $m_{(xy)}^i$ is the coding-modulation scheme with the highest transmission rate assuming node $a$ transmits with power $\tilde{p}_{c,i}.$
    \item[(v)] Compute the virtual local weight, under each critical power level $\tilde{p}_{c,i}:$
    $$\tilde V_{ab}(\tilde{p}_{c,i}) = \sum_{y\in {\cal N}^1(i), (xy)\in {\mathcal E}} \tilde r_{xy}(\tilde{p}_{c,i}) q_{xy}$$
    for $i=0, 1, \cdots,$ where $q_{xy}$ is the queue length at the beginning of the super time slot.
\end{itemize}

\item[(5)] {\bf Power-level selection:} Let $Z_{ab}$ be the normalization factor defined as
$$Z_{ab} = \sum_{i} \left(e^{-\frac{\epsilon \tilde{p}_{c,i}}{K_{t}}} - e^{-\frac{\epsilon \tilde{p}_{c, i+1}}{K_{t}}}\right)e^{\frac{\tilde V_{ab}(\tilde{p}_{c,i})}{K_{t}}}.$$
Node $a$ first selects a power interval $[\tilde{p}_{c,i}, \tilde{p}_{c, i+1})$ with following probability:
\begin{equation}\label{eq: algprob}
    \Pr\left(\tilde{p}_{ab} \in [\tilde{p}_{c,i}, \tilde{p}_{c, i+1})\right)  = \frac{1}{Z_{ab}}  \left(e^{-\frac{\epsilon \tilde{p}_{c,i}}{K_{t}}} - e^{-\frac{\epsilon \tilde{p}_{c, i+1}}{K_{t}}}\right) e^{\frac{\tilde V_{ab}(\tilde{p}_{c,i})}{K_{t}}}.
\end{equation}
Suppose the interval $[\tilde{p}_{c,k}, \tilde{p}_{c, k+1})$ is selected, then node $a$ randomly selects a virtual power level $\tilde{p}_{ab}(t)$ according to the following probability density function (pdf):
\begin{equation}\label{eq: algdensity}
     f_{ab}(p | p\in [\tilde{p}_{c,k}, \tilde{p}_{c, k+1})) = \frac{\epsilon}{K_{t}} e^{-\frac{\epsilon p}{K_{t}}}\left(e^{-\frac{\epsilon \tilde{p}_{c,k}}{K_{t}}} - e^{-\frac{\epsilon \tilde{p}_{c, k+1}}{K_{t}}}\right)^{-1},
\end{equation}
which can be done by using the inverse transform sampling method.

\item[(6)] {\bf Information exchange:} If the virtual power of a node has changed, i.e., $\tilde{p}_{ab}(t-1)\not=\tilde{p}_{ab}(t),$ then node $a$ broadcasts $\tilde{p}_{ab}(t)$ to all its one-hop neighbors. Each neighbor $y$ computes the virtual partial-interference-plus-noise of $(xy),$ i.e., $\tilde{\Upsilon}_{xy}(t).$ If $\tilde{\Upsilon}_{xy}(t-1)\not=\tilde{\Upsilon}_{xy}(t),$ node $y$ broadcasts $\tilde{\Upsilon}_{xy}(t)$ to all its one-hop neighbors.

\item[(7)] {\bf Update real transmit power:} At the end of a super time slot, node $a$ updates its transmit power for link $(ab)$ to be $p_{ab}=\tilde{p}_{ab}(T).$ Node $b$ then measures the actual SINR  $\gamma_{ab}$ and reports to node $a.$ Node $a$ selects the coding-modulation scheme $m_{ab}$ with the highest rate among those $m$ such that $\gamma_{(ab),m}\leq \gamma_{ab}.$ Packets of flow $(ab)$ are transmitted with power $p_{ab}$ and coding-modulation scheme $m_{ab}.$ Note that the real transmit powers are updated only once every $T$ time slots.
\end{enumerate}

We now present a simple example to show how the algorithm works.

{\bf Example:}\\
Consider the wireless network depicted in Fig. \ref{fig: eg1}. There are 4 links in the network.
Assume the channel gain, background noises and queue lengths are:
$g_{ab} = g_{cd} = g_{ef} = 1,$ $g_{cb} = g_{cf} = g_{ed} = g_{ad} = \frac{1}{4},$ $n_b = n_d = n_f = 1,$ and $q_{ab} = 10, q_{cd} = 100, q_{ef}=10.$ The virtual power level of the links in the previous time slot are $\tilde{p}_{ab}(t-1) = 15, \tilde{p}_{cd}(t-1) = 0, \tilde{p}_{ef}(t-1) =10.$ Further, we assume that there are two feasible coding-modulation schemes for each link: BPSK with rate 1, and QPSK with rate 2, and the SINR requirement for the modulations are 4 and 8, respectively.


\begin{figure}[hbt!]
\begin{center}
  \includegraphics[width=2in]{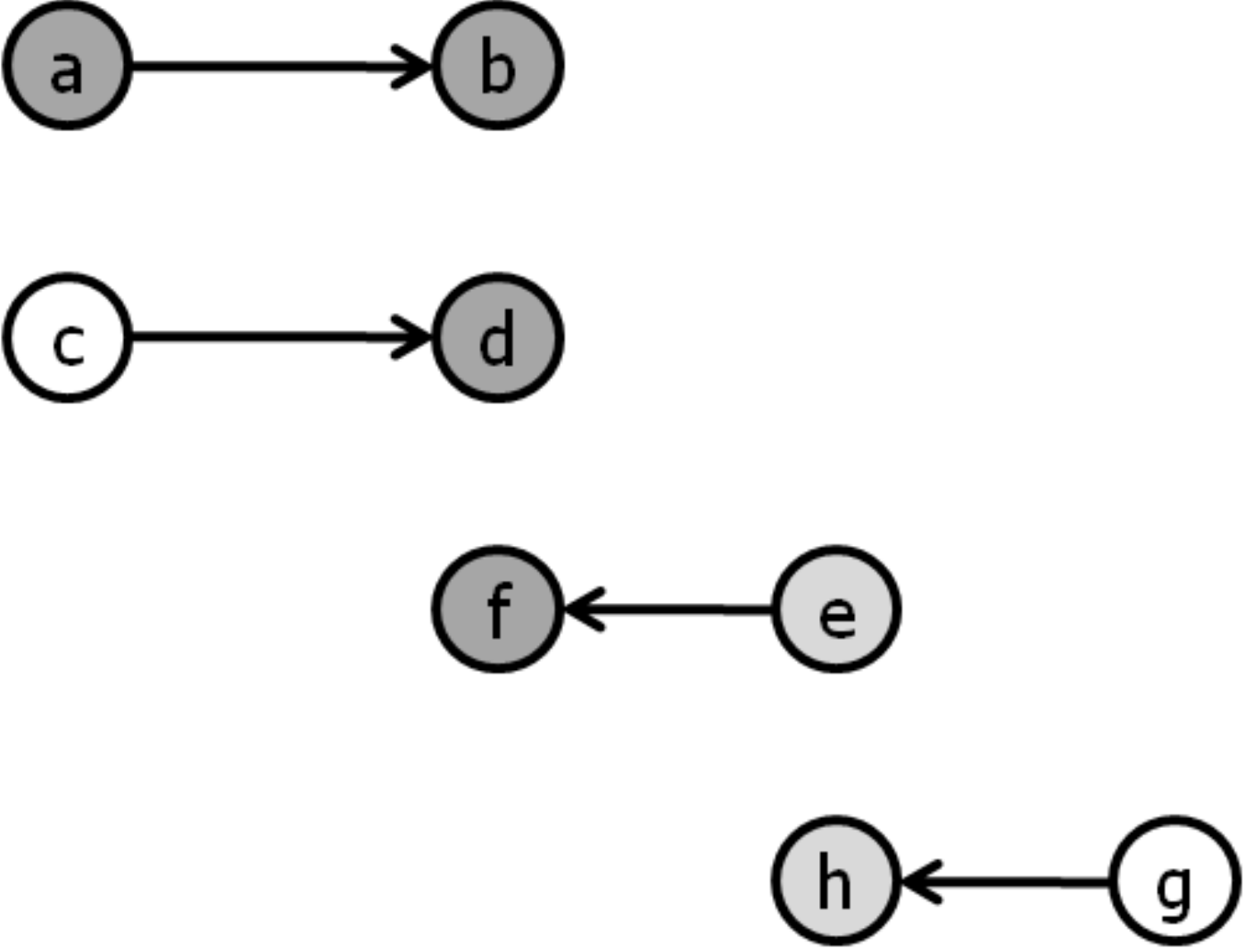}\\
  \end{center}\caption{A simple example}\label{fig: eg1}
\end{figure}

In this example, we focus on link $(cd).$ Assume that the set of $c$'s one-hop neighbors is ${\cal N}^1(c) = \{a, b, d, f\},$ and the set of two-hop neighbors is ${\cal N}^2(c) =\{e, h\}.$

Under this neighborhood structure, if $c$ changes its power, it will change three links' virtual SINR, and their virtual rates, i.e., links $(ab), (cd)$ and $(ef).$

In the algorithm, node $c$ randomly select a power level based on its interference to the neighboring links with each feasible power level.

\begin{enumerate}

\item[(1)]{\bf Select decision set:} Assume that the number of control slots is $W = 5,$ and the backoff time generated by the transmitters are $$T_a=2, T_c = 0, T_e = 3, T_g = 1.$$
Then, $c$ broadcasts an INTENT message at control slot $1$, which is received by $a$ and $e.$ Node $g$ will ignore this INTENT message even if it can receive it, because $g$ is not within the two-hop range of $c.$ In control slot $2$, node $g$ broadcasts an INTENT message, which is received by node $e.$ Since both transmitters $a$ and $e$ receive the INTENT message sent by $c,$ they are not in the decision set. And the decision set is $\{c, g\}.$

\emph{Remark:} Note that $c$'s transmit power affects the virtual rate of links $(ab),$ $(cd),$ and $(ef),$ while $g$'s power affects the virtual rate of link $(gh)$ only. Thus, no link's virtual rate is affected by both of $c$ and $g.$

\item[(2)] {\bf Information exchange:} Suppose the power level of link $(ab)$ has changed in time slot $t-1,$ node $a$ then has broadcast $p_{ab}(t-1)$ to all its two-hop neighbors. Node $c$ has received this message, which shows that $c$ always knows the power level of $(ab)$ and $ef.$

\item[(3)] {\bf Link selection:} Each transmitter only has one link, so if a transmitter is selected to be in the decision set, its outgoing link will be selected.

\item[(4)]{\bf Critical power computation:} Node $c$ knows the virtual partial-interference-plus-noise experienced by links $(ab),(cd), (ef):$
\begin{align*}
  \tilde{\Upsilon}_{ab} (t-1) &= n_b + \tilde{p}_{cd}g_{cb} = 1 \\
  \tilde{\Upsilon}_{cd} (t-1)& = n_d + \tilde{p}_{ab}g_{ad} + \tilde{p}_{ef}g_{ed} = 7.25\\
  \tilde{\Upsilon}_{ef} (t-1) &=n_f + \tilde{p}_{cd}g_{cf} = 1
\end{align*}
Then, node $c$ can estimate its impact on links $(ab)$ and $(ef),$ when it varies transmit power from 0 to $p^{\rm max}.$

Fig. \ref{fig: cp} illustrates this impact. We can see from the figure that there are 5 critical power levels besides 0 and $p^{\rm max},$ which are $1, 3.5, 6, 11, \text{and } 29.$ Take critical power level $1$ for example, it means that if the power of link $(cd)$ is greater than $1,$ then the virtual SINR of link $(ef)$ will be below $8,$ and link $(ef)$ will only be able to use BPSK.
\begin{figure}[!hbt]
\begin{center}
  \includegraphics[width=3in]{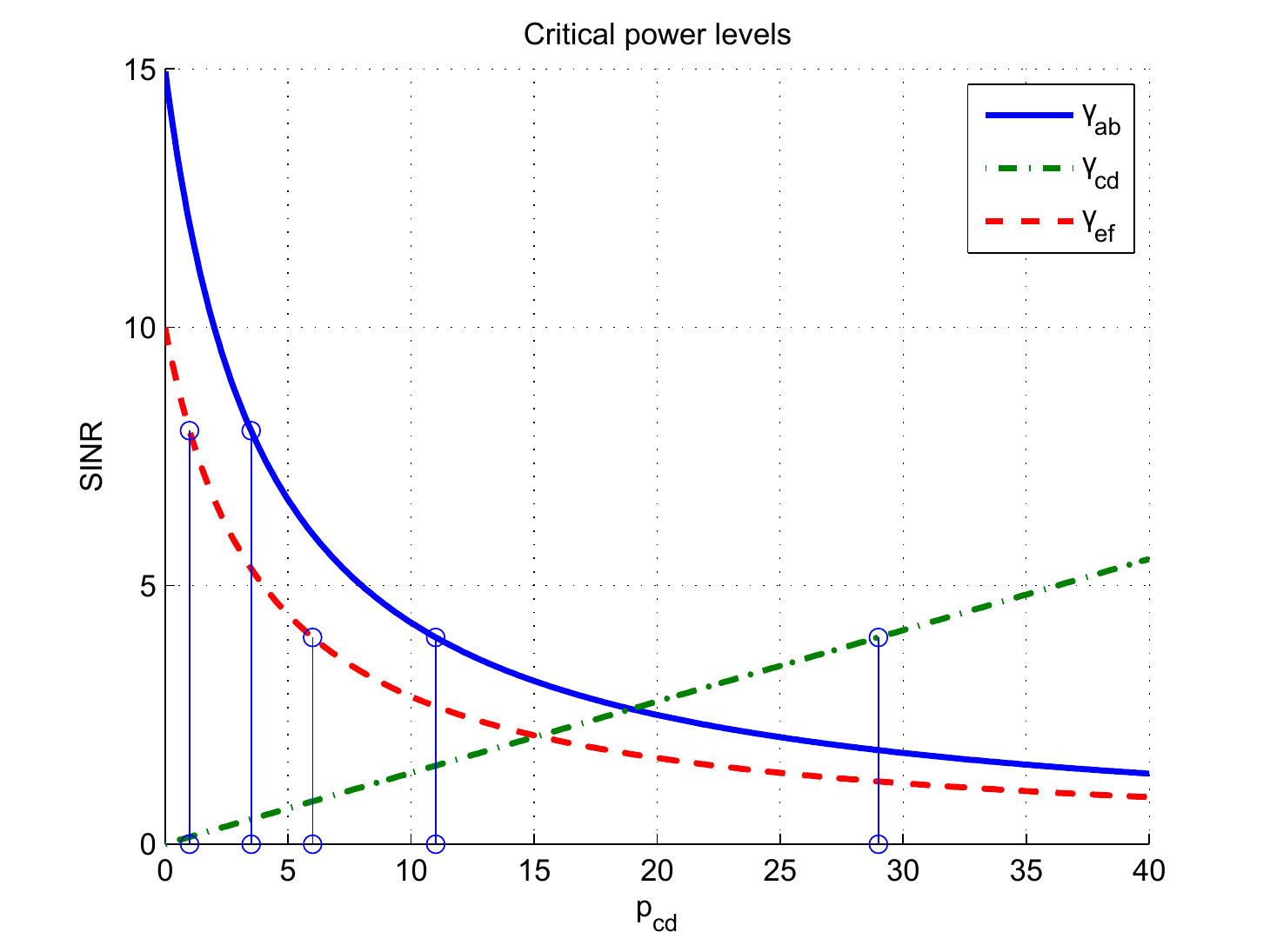}\\
  \end{center}\caption{Critical power levels}\label{fig: cp}
\end{figure}

\item[(5)] {\bf Virtual rates computation:} Now, $c$ knows the critical power levels and the coding-modulation schemes corresponding to each interval between the critical power levels. Thus, $c$ can calculate $\tilde V_{cd}(\tilde{p}_{cd})$ accordingly, which is shown in Table \ref{tb: example}. Given these virtual rates, node $c$ then samples a power level according to the distribution in equalities (\ref{eq: algprob}) and (\ref{eq: algdensity}).
\begin{table*}[t]
\centering
\caption{Critical power level and the resulted virtual rate}
\begin{tabular}{|c|c|c|c|c|c|c|} \hline
  $p_{cd}$ & $[0,1)$ & $[1,3.5)$  & $[3.5,6)$ & $[6,11)$ & $[11,29)$ & $[29,40]$ \\\hline
  $(\tilde r_{ab}(p_{cd}), \tilde r_{cd}(p_{cd}), \tilde r_{ef}(p_{cd}))$ & $(2,0,2)$&$(2,0,1)$&$(1,0,1)$&$(1,0,0)$&$(0,0,0)$&$(0,1,0)$ \\\hline
  $\tilde V_{cd}(p_{cd})$ & $40$ & $30$ & $20$ & $10$ & $0$ & $100$\\\hline
\end{tabular}
\label{tb: example}
\end{table*}
\end{enumerate}
\subsection{Analysis}
In classical Gibbs sampler, the state of each link is updated in a sequential manner. In contrast, the Gibbs sampler used in our algorithm is parallelized and distributed, which leverages the distributed characteristic of wireless networks. In the following lemma, we prove that our algorithm generates a sequence of power configurations which form a Markov chain with some desired stationary density.
\begin{lemma}\label{lem: Stationary}
For a fixed temperature, i.e., without updating the temperature in the power control algorithm, and fixed queue lengths, the sequence of the power configurations, $\mathbf p(t)$, generated by the power control algorithm, forms a Markov chain with the stationary density:
$$\pi_{K_t}(\mathbf p) = \frac{1}{Z(K_{t})}e^{\frac{\tilde{V} (\mathbf p) -\epsilon \sum_{(ab)\in {\mathcal E}}p_{ab}}{K_{t}} },$$ where
$Z(K_{t}) = \int_{{\mathbf p} \in {\mathcal P}} e^{\frac{\tilde{V} (\mathbf p)- \epsilon \sum_{(ab)\in {\mathcal E}}p_{ab}}{K_{t}}} d\mathbf p$ is unknown normalization factor.
\end{lemma}
\begin{proof}
The proof is presented Appendix \ref{sec: stationary}.
\end{proof}
In our algorithm, the virtual powers are updated using an annealed Gibbs sampler. Assuming the queues are fixed, the following theorem states that, with fixed queue length, the power configurations converge to the optimal solution to (\ref{eq: obj}) as $t$ goes to infinity.
Let $U(\tilde{\mathbf p}) = \tilde V(\tilde{\mathbf p}) - \epsilon \sum_{(ab)\in \mathcal E} \tilde{p}_{ab},$ and
\begin{align*}
U^* &= \max_{\tilde{\mathbf p}\in \mathcal P}U(\tilde{\mathbf p}), U_* = \min_{\tilde{\mathbf p}\in \mathcal P} U(\tilde{\mathbf p}), \Delta = U^* - U_*
\end{align*}
\begin{thm}
Let $K_0 = 2n \Delta.$ Assume $q_{ab}$ are fixed and $\mathcal P_{\delta}^*$ is the set of power configurations such that for any $\tilde{\mathbf p}\in  \mathcal P_{\delta}^*,$
\begin{eqnarray*}
&\sum_{(ab) \in {\cal E}}{q_{ab}\tilde{r}_{ab}(\tilde{\mathbf p})} -\epsilon \sum_{(ab) \in {\cal E}} \tilde{p}_{ab}\\
&\geq (1-\delta) \max_{\tilde{\mathbf p}} \sum_{(ab) \in {\cal E}}{q_{ab}\tilde{r}_{ab}(\tilde{\mathbf p})}-\epsilon \sum_{(ab) \in {\cal E}} \tilde{p}_{ab}.\end{eqnarray*}  Then given any $\delta>0,$ $\varepsilon>0,$ and starting from any initial power configuration, $\tilde{\mathbf p}_0 \in \mathcal P,$ there is an $N\in \mathbb N$ such that if
\begin{equation}
    K_t = \left\{\begin{array}{c c}
                  \frac{K_0}{\log(2+t)} & 0<t<N, \\
                  \frac{K_0}{\log(2+N)} & N\leq t
                \end{array}   \right.
\end{equation}
we have
$$\lim_{t\rightarrow \infty} \int_{\tilde{\mathbf p}_{\delta}^*} P(t, \tilde{\mathbf p} | 0, \mathbf p_0) d\tilde{\mathbf p} = 1 - \varepsilon.$$ \label{thm: Optimal}
\end{thm}
\begin{proof}
The proof is presented in Appendix \ref{sec: proofOptimal}.
The proof of the theorem follows the idea in \cite{GemGem_90}. However, in our algorithm, the decision set is randomly generated instead of predetermined, and the Markov chain has a continuous state space instead of a discrete state space, so the convergence of the annealed Gibbs sampling is not guaranteed. The proof therefore is a nontrivial extension.
\end{proof}

{\em Remark 1:} The theorem requires that the queue lengths are fixed during the annealing, which is the reason the algorithm uses the queue lengths at the beginning of a super time slot for the entire super time slot.

{\em Remark 2:} In the algorithm, we replace the interference from non-neighbor nodes with upper bound $\xi.$ Therefore, when node $a$ changes its transmit power to node $b$ to $\tilde{p}_{ab}(T)$ at the end of a super time slot, the actual rate $r_{ab}$ is at least $\tilde{r}_{ab}(T),$ because the actual interference is smaller than that in the virtual rate computation. Further when the neighborhood is chosen to be large enough, i.e., $\xi$ is small, the optimal configuration based on virtual rate is close to the optimal configuration with global interference. But a large neighborhood increases both the computation and communication complexities.

\section{Simulations}
In this section, we use simulations to evaluate the performance of the proposed algorithm, which is SINR-based, with CSMA-based algorithm and Q-CSMA\cite{NiTanSri_10}. {\color{black}The CSMA-based algorithm used in the simulation is an approximation of the traditional CSMA/CA with RTS/CTS algorithm. It is implemented as the following. In each time slot, one link is uniformly randomly selected to transmit. Then the links whose receiver is in the carrier sensing range of the selected transmitter are marked and cannot transmit in the time slot. Then another link in the rest of the links is uniformly randomly select to transmit. Repeat this procedure until there is no more link to select. Thought there is no RTS/CTS in the implementation, this algorithm capture the essence of the CSMA algorithm and has similar performance.} In the simulations, we assume the channel attenuation over a distance $l$ is $l^{-3.5},$ where the path loss exponent is chosen to be $3.5.$ All channels are assumed to be AWGN channels. The transmit power can be continuously adjusted from $[0, 100].$ The rate options for each link are $6,$ $9,$ $12,$ $18,$ $24,$ $36,$ $48,$ and $54$ Mbps, which are the eight rate options available in 802.11g\cite{MhaPapBac_07}. The system is time-slotted, and each time slot is 1 ms. We assume each packet is of size $1,500$ bytes, i.e., 12 Kbits.  So when the link rate is $54$ Mbps, $4.5$ packets can be transmitted in one time slot. {\color{black} Each super time slot consists of $T=50$ time slots. $\alpha$ is equal to $100^{-3.5},$ which is the threshold of the channel gain between two neighboring nodes.}

\subsection{A Ring Network}
Consider a ring network consisting of $9$ directed links, as shown in Fig \ref{fig: ringnetwork}. Each node in the network has one transceiver. The length of each link is $20$ meters. We assume the carrier sensing range is $40$ meters, which is slightly larger than then distance between two nodes that are two-hop away. 

\begin{figure}[!hbt]
\begin{center}
  \includegraphics[width=3.5in]{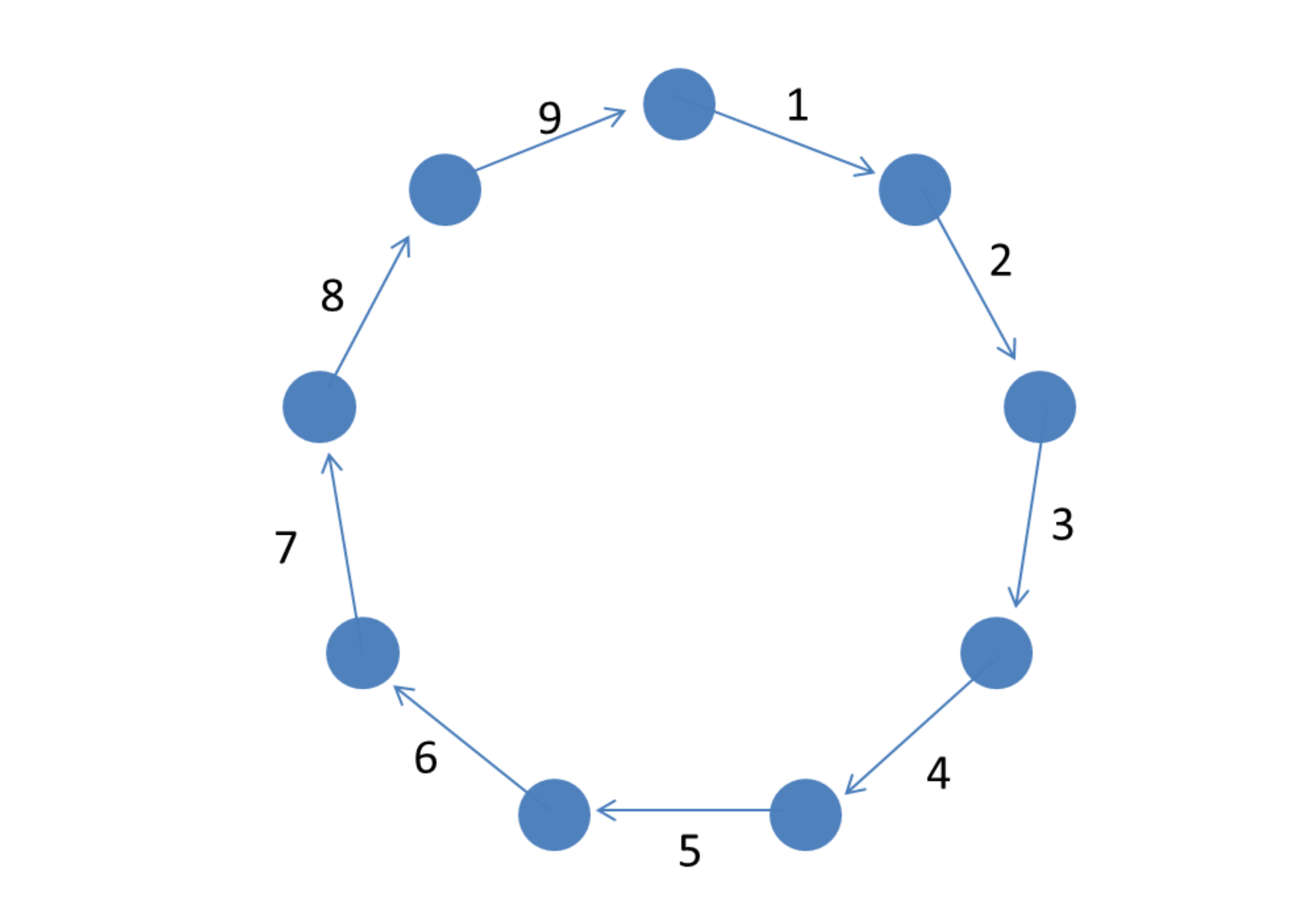}\\
  \end{center}\caption{A ring network containing 9 links}\label{fig: ringnetwork}
\end{figure}
The arrival process is the same as the one described in \cite{NiTanSri_10}. Namely, at time slot $t$, one packet arrives at the transmitters of links $(t \mod 9)$ and $((t+4) \mod 9);$ additionally, with probability $\rho,$ one packets arrives at each transmitter. Hence, the overall arrival rate is $2+9\rho$ packets per time slot. In the simulation, we varied $\rho$ from $0.10$ to $0.26,$ which corresponds to varying the overall arrival rate from $2.9$ to $4.34$ (packets/time slot).

For each value of $\rho,$ we run each simulation for $10^5$ time slots. Figure \ref{fig: simRingNet} shows the mean of the sum queue length in the network. It shows that the sum queue length grows unbounded under the CSMA-based algorithm when $\rho\geq 0.11$ (i.e., overall arrival rate $\geq 2.99$). In other words, the network is unstable under CSMA algorithm for $\rho \geq 0.11.$ On the other hand, our algorithm stabilizes the network for any $\rho\leq 0.25,$ with a corresponding overall arrival rate equal to $4.25.$ Hence, our algorithm increases the throughput by $47\%$ comparing to the CSMA-based algorithm. We can also see that, Q-CSMA, which is throughput optimal under the collision interference model, has similar throughput as the CSMA (around 3). The implementation details of Q-CSMA can be found in \cite{NiTanSri_10}. 

\begin{figure}[!hbt]
\begin{center}
  \includegraphics[width=3.5in]{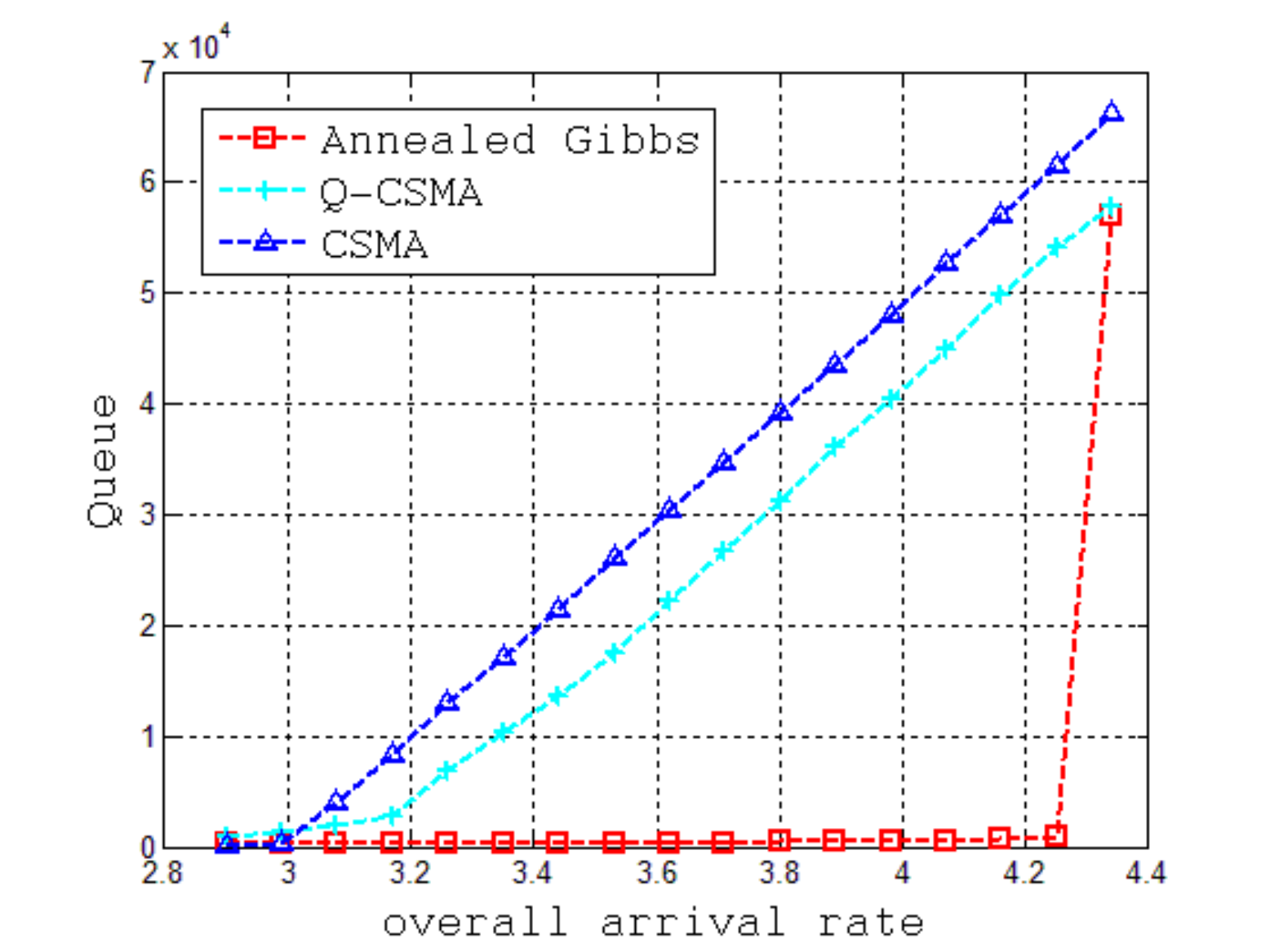}\\
  \end{center}\caption{Average queue length in the ring network}\label{fig: simRingNet}
\end{figure}

\subsection{A Random Network}
In this simulation, we randomly place $200$ links, each with length $20$ meters, in a $1000\times1000$ $\hbox{meter}^2$ two-dimensional torus. The carrier sensing range is set to be $200$ meters, which corresponds to a sensing threshold of $-91$ dBm. We assume Poisson arrivals for each link, and the arrival rate is the same for all links.

\begin{figure}[!hbt]
\begin{center}
  \includegraphics[width=3.5in]{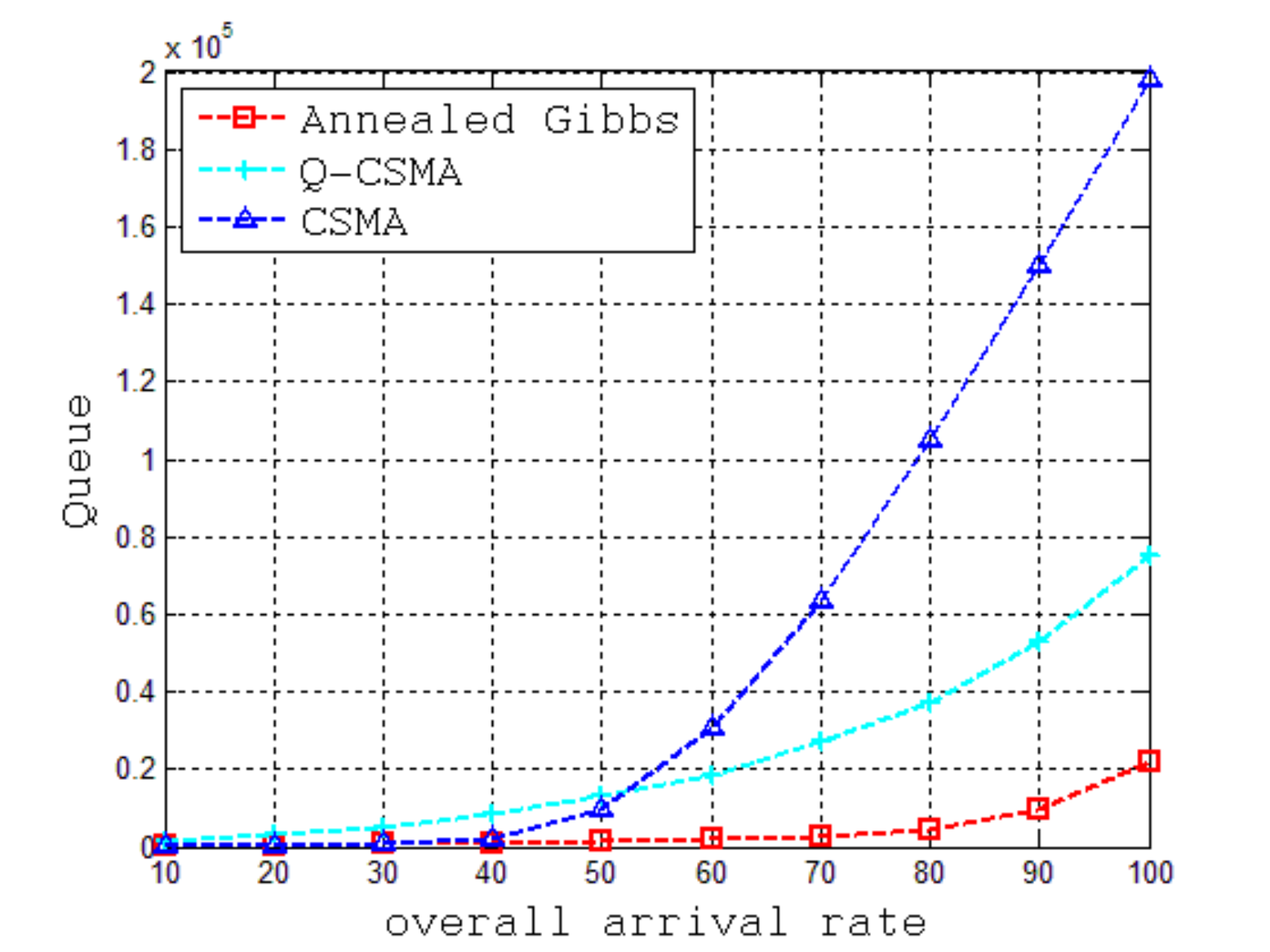}\\
  \end{center}\caption{Average queue length in the random network}\label{fig: simRandNet}
\end{figure}

For each arrival rate, the simulations is run for $10^4$ time slots. Fig \ref{fig: simRandNet} illustrates the time average value of the total queue length in the network under arrival rates and different rate-control algorithms. We observe that the supportable throughput is $70$ packets/slot under the proposed algorithm, $40$ packets/slot under the CSMA algorithm. Our algorithm therefore achieves a $75\%$ throughput gain. Comparing to Q-CSMA, we can see that our algorithm has a much smaller queue length and hence has a much smaller delay.

From these simulation results, we observed that our algorithm significantly outperforms the CSMA-based algorithm and Q-CSMA algorithm, which confirms the importance of adapting transmit powers and coding-modulation schemes in wireless network to increase the network throughput.

\section{Discussion and Conclusion}
{\color{black} We remark that there are several important parameters that should be carefully tuned according to the network configuration to optimize the performance.

First, there is a trade-off between complexity and performance in selecting $\alpha.$ If $\alpha$ is equal to zero, all nodes are one-hop neighbors of each other, then the optimal virtual power configuration obtained by our algorithm $\tilde{\mathbf p}$ is the same as optimal power configuration ${\mathbf p}.$ However, the number of one-hop neighbors of each node should be bounded so that the signaling overhead is acceptable in practice. To bound the number of one-hop neighbors, $\alpha$ should not be too small. On the other hand, a small $\alpha$ is preferred to keep the virtual power configuration to be close to the real power configuration. Therefore, $\alpha$ should be carefully determined based on the real network configuration.

Furthermore, the Markov chain converges to the stationary distribution and the probability of being in the optimal power configuration converges to one only when $T$ goes to infinity. In practice, we cannot choose $T$ to be infinity or too large because the algorithm will response very slowly to queue change and will lead to very bad delay performance. So $T$ also should be carefully chosen in practice. Selecting these parameters to optimize the network performance is an important issue of implementing the proposed algorithm in practice. The problem, however, is complicated and requires further investigation, so left as future topics of our research.}

In summary, we developed a distributed power control and coding-modulation adaptation algorithm using annealed Gibbs sampling, which achieves throughput optimality in an arbitrary network topology. The power update policy emulates a Gibbs sampler over a Markov chain with a continuous state space. Simulated annealing is exploited in the algorithm to speed up the convergence of the algorithm to the optimal power and coding-modulation configuration. Simulation results demonstrated the superior performance of the proposed algorithm.

\appendices
\section{Proof of Lemma \ref{lem: Stationary}}\label{sec: stationary}
We begin the proof with the following lemma, which states that our algorithm simulates a time homogeneous Markov chain.
\begin{lemma}\label{lem: timeHomoMC}
For fixed temperature and queue lengths, the power configurations generated by the power control algorithm form a homogeneous Markov chain with state space ${\mathcal P}.$
\end{lemma}
\begin{proof}
Let ${\mathbf p}$ denote the current power configuration, and $\mathbf{\hat{p}}$ denote the power configuration generated by the algorithm.

First, by observing the procedure of the generation of the decision set, it is clear that the decision set ${\mathcal D}$ is independent of the power configuration ${\mathbf p}.$ Moreover, in the link selection stage, the links are selected based on ${\mathcal D}$ and ${\mathbf p}.$ Thus, ${\mathcal D}_l$ depends on ${\mathbf p}$ only. Second, for each link $(ab) \in {\mathcal D}_l,$ the new power is sampled from the density function which is determined by ${\mathbf p}$ and ${\mathbf{\hat{p}}},$ while is independent of the earlier power configurations than ${\mathbf p}.$ The claim then follows.
\end{proof}

Knowing that our algorithm simulates a homogeneous Markov chain, the following lemma gives us the transition kernel density.
\begin{lemma}\label{lem: transition}
Suppose link $(ab)$ is selected to update its transmit power, then its power is randomly selected according to the following density:
$$g(p_{ab} | (ab)\in \mathcal D_l, p_{xy}: y\in \mathcal N^1(a)) = \frac{1}{Z_{ab}(K_{t})}e^{\frac{\tilde{V}_{ab}(p_{ab}) - \epsilon p_{ab}}{K_{t}}},$$
where $$Z_{ab}(K_{t}) = \int_0^{p_a^{\rm max}}e^{\frac{\tilde{V}_{ab}(p) - \epsilon p}{K_{t}}} dp$$ is a normalization constant independent of $p_{ab}.$
\end{lemma}
\begin{proof}
The proof is presented in Appendix \ref{sec: proofLemTransition}.
\end{proof}
It can be easily verified that all the power configurations communicate with the zero power configuration, in which the transmit power of all transmitters are zero. Also,the zero power configuration has a self-loop, which indicates the Markov chain is irreducible and aperiodic and thus ergodic. In the following two lemmas, we will show that conditioning on any link decision set, the detailed balance equations holds, which leads to the conclusion of the lemma.
\begin{lemma}\label{lem: reachablity}
Let $$P(\mathbf p, \mathbf{\hat{p}}) = P(\mathbf p(t+1) = \mathbf{\hat{p}} | \mathbf p(t)= \mathbf p)$$ be the transition kernel probability density, and $$P(\mathbf p, \mathbf{\hat{p}} | \mathcal D_l) = P(\mathbf p(t+1) = \mathbf{\hat{p}} | \mathcal D_l, \mathbf p(t) = \mathbf p).$$
For two power configurations ${\mathbf p}, \mathbf{\hat{p}} \in {\mathcal P},$ if
$$\Pr(\mathcal D_l| \mathbf p(t) = \mathbf p) > 0, \textrm{ and } P(\mathbf p, \mathbf{\hat{p}}| \mathcal D_l) >0,$$
then,
$$\Pr(\mathcal D_l| \mathbf p(t) = \mathbf{\hat{p}}) = \Pr(\mathcal D_l| \mathbf p(t) = \mathbf p), \textrm{ and }P(\mathbf{\hat{p}}, \mathbf p | \mathcal D_l) >0.$$

In other words, if $\mathbf{\hat{p}}$ is reachable from $\mathbf p$ in one transition with a link decision set $\mathcal D_l,$ then $\mathbf p$ is reachable from $\mathbf{\hat{p}}$ in one transition with the same link decision set.
\end{lemma}
\begin{proof}
The proof is presented in Appendix \ref{sec: proofLemReachability}.
\end{proof}

\begin{lemma} \label{lem: conditionalreversible}
For any link in the link decision set $\mathcal D_l,$ if $$\Pr(\mathcal D_l | \mathbf p(t) = \mathbf p)>0,$$ then
$$\pi(\mathbf p) P(\mathbf p, \mathbf{\hat{p}} | \mathcal D_l) = \pi(\mathbf{\hat{p}}) P(\mathbf{\hat{p}}, \mathbf p |\mathcal D_l).$$
\end{lemma}
\begin{proof}
The proof is presented in Appendix \ref{sec: conditionalReversible}.
\end{proof}
Finally, we have
\begin{align*}
  \pi(\mathbf p) P(\mathbf p, \mathbf{\hat{p}}) =& \pi(\mathbf p) \sum_{\mathcal D_l} P(\mathbf p, \mathbf{\hat{p}}|\mathcal D_l) \Pr(\mathcal D_l | \mathbf p)  \\
  =& \pi(\mathbf{\hat{p}}) \sum_{\mathcal D_l} P(\mathbf{\hat{p}}, \mathbf p|\mathcal D_l) \Pr(\mathcal D_l | \mathbf p) \\
  =& \pi(\mathbf{\hat{p}}) \sum_{\mathcal D_l} P(\mathbf{\hat{p}}, \mathbf p|\mathcal D_l) \Pr(\mathcal D_l | \mathbf{\hat{p}}) \\
  =& \pi(\mathbf{\hat{p}}) P(\mathbf{\hat{p}}, \mathbf p),
\end{align*}
where the second equality holds by Lemma \ref{lem: conditionalreversible}, and the third equality holds by Lemma \ref{lem: reachablity}.
We then concludes the lemma.


\section{Proof of Lemma \ref{lem: transition}}\label{sec: proofLemTransition}
\begin{proof}
First, we arrange the critical power levels of $(ab),$ $$\{p_c(ab, xy), \forall (xy) \quad {\textrm s.t. }\quad y\in {\cal N}^1(a)\}$$ in ascending order, denoted by
     $$0=p_{c,0}<p_{c,2}<\ldots <p_{c,m} = p^{\rm max}.$$

Recall that the power level of link $(ab)$ is generated using the following procedure:

Node $a$ first selects a power interval $[p_{c,i}, p_{c, i+1})$ with following probability:
$$\Pr\left(p \in [p_{c,i}, p_{c, i+1})\right)  = \frac{1}{Z_{ab}}\left(e^{-\frac{\epsilon p_{c,i}}{K_{t}}} - e^{-\frac{\epsilon p_{c, i+1}}{K_{t}}}\right) e^{\frac{\tilde V_{ab}(p_{c,i})}{K_{t}}},$$
where $$Z_{ab} = \sum_{i=0}^{m-1}  \left(e^{-\frac{\epsilon p_{c,i}}{K_{t}}} - e^{-\frac{\epsilon p_{c, i+1}}{K_{t}}}\right) e^{\frac{\tilde{V}_{ab}(p_{c,i})}{K_{t}}}.$$
Given the interval $[p_{c,i}, p_{c, i+1})$ is selected, then $a$ randomly select the power level $p_{ab}(t)$ according to the following probability density function(pdf):
$$f_{ab}(p | p\in [p_{c,i}, p_{c, i+1})) = \frac{\epsilon}{K_{t}} e^{-\frac{\epsilon p}{K_{t}}}\left(e^{-\frac{\epsilon p_{c,i}}{K_{t}}} - e^{-\frac{\epsilon p_{c, i+1}}{K_{t}}}\right)^{-1}.$$

For any power level $p\in [0, p_a^{\rm max}],$ let us consider the probability density. it must in an interval $[p_{c,i}, p_{c,i+1})$ for some $i, i\in\{0,1,\ldots, m-1\}.$ Hence, $p$ is selected according to the following density:

\begin{align*}
  f_{ab}(p)  &= f_{ab}( p | p\in [p_{c,i}, p_{c, i+1}))\Pr\left(p \in [p_{c,i}, p_{c, i+1})\right) \\
   &= \frac{1}{Z_{ab}}\frac{\epsilon }{K_{t}} e^{-\frac{\epsilon p}{K_{t}}}  e^{\frac{\tilde{V}_{ab}(p_{c,i})}{K_{t}}}\\
\end{align*}

We then need to show when $p\in [p_{c,i}, p_{c,i+1}),$
$$ \frac{e^{\frac{\tilde{V}_{ab}(p_{ab}) - \epsilon p_{ab}}{K_{t}}}}{\int_0^{p_a^{\rm max}}e^{\frac{\tilde{V}_{ab}(p) - \epsilon p}{K_{t}}} dp} = \frac{1}{Z_{ab}}\frac{\epsilon }{K_{t}} e^{-\frac{\epsilon p_{ab}}{K_{t}}}  e^{\frac{\tilde{V}_{ab}(p_{c,i})}{K_{t}}},$$
which is equivalent to show that
$$\int_0^{p_a^{\rm max}}e^{\frac{\tilde{V}_{ab}(p) - \epsilon p}{K_{t}}} dp = Z_{ab}\frac{K_{t}}{\epsilon }.$$

We simplify the LHS term by computing the integral in each interval $[p_{c,i}, p_{c,i+1})$.
$$\int_{p_{c,i}}^{p_{c,i+1}} e^{\frac{\tilde{V}_{ab}(p) - \epsilon p}{K_{t}}} dp.$$
Notice that, by the definition of critical power level, the modulation of link $(ab)$ and all the links $(xy),\, y\in {\cal N}^1(a),$ will not change if $p_{ab}$ varies between two adjacent critical power levels $p_{c,i}$ and $p_{c,i+1}.$ Thus, $$\tilde{V}_{ab}(p_{ab}) = \sum_{(xy): y\in {\cal N}^1(a)} \tilde r_{xy} q_{xy} = \tilde{V}_{ab}(p_{c,i})$$ is a constant, when $p\in [p_{c,i}, p_{c,i+1}).$
Hence,
\begin{align*}
  &\int_{p_{c,i}}^{p_{c,i+1}} e^{\frac{\tilde{V}_{ab}(p) - \epsilon p}{K_{t}}} dp\\
  =& e^{\frac{\tilde{V}_{ab}(p_{c,i})}{K_{t}}}\int_{p_{c,i}}^{p_{c,i+1}}  e^{\frac{- \epsilon p}{K_{t}}} dp\\
  =& -\frac{K_{t}}{\epsilon} e^{\frac{\tilde{V}_{ab}(p_{c,i})}{K_{t}}}  e^{\frac{- \epsilon p}{K_{t}}}|_{p_{c,i}}^{p_{c,i+1}}\\
  =& \frac{K_{t}}{\epsilon} e^{\frac{\tilde{V}_{ab}(p_{c,i})}{K_{t}}} \left(e^{-\frac{p_{c,i}}{K_{t}}} - e^{-\frac{p_{c,i+1}}{K_{t}}}\right)
\end{align*}

Therefore, we have
\begin{align*}
  &\int_0^{p_a^{\rm max}}e^{\frac{\tilde{V}_{ab}(p) - \epsilon p}{K_{t}}} dp\\
  =& \sum_{i=0}^{m-1} \int_{p_{c,i}}^{p_{c,i+1}} e^{\frac{\tilde{V}_{ab}(p) - \epsilon p}{K_{t}}} dp\\
  =& \sum_{i=0}^{m-1} \frac{K_{t}}{\epsilon} e^{\frac{\tilde{V}_{ab}(p_{c,i})}{K_{t}}} \left(e^{-\frac{p_{c,i}}{K_{t}}} - e^{-\frac{p_{c,i+1}}{K_{t}}}\right)\\
  =& \frac{K_{t}}{\epsilon}Z_{ab}
\end{align*}

The Lemma then follows.
\end{proof}

\section{Proof of Lemma \ref{lem: reachablity}}\label{sec: proofLemReachability}
\begin{proof}
Suppose $\mathcal D_l$ is a link decision set with $\Pr(\mathcal D_l | \mathbf p(t) = \mathbf p) >0.$ A power configuration $\mathbf{\hat p}$ is reachable in one time slot from $\mathbf p,$ i.e., $P(\mathbf p, \mathbf{\hat{p}}| \mathcal D_l) >0.$
Assume that $\mathcal D$ is the decision set corresponding to $D_l,$ we then have
$$\Pr(\mathcal D| \mathbf p(t) = \mathbf p)>0.$$
Remember that the decision set is generated independently of the power configuration, which means
$$\Pr(\mathcal D| \mathbf p(t) = \mathbf p) = \Pr(\mathcal D | \mathbf p(t) = \mathbf{\hat{p}})>0.$$

Now we consider an arbitrary link $(ab)\in \mathcal D_l,$ which implies that $a\in \mathcal D.$ There are two different cases:
\begin{enumerate}
  \item If node $a$ has only one outgoing link, then the event $(ab)\in \mathcal D_l$ and the event $a\in \mathcal D$ are equivalent.  Thus,
      \begin{align*}
        \Pr((ab)\in \mathcal D_l | \mathbf p(t) = \mathbf p) &= \Pr(a\in \mathcal D| \mathbf p(t) = \mathbf p)\\
        &= \Pr(a\in \mathcal D| \mathbf p(t) = \mathbf{\hat p})\\
        &= \Pr((ab)\in \mathcal D_l | \mathbf p(t) = \mathbf{\hat p})\\
      \end{align*}
  \item If node $a$ has more than one outgoing link, since link $(ab)\in D_l$ when the power configuration is $\mathbf p,$ then no other outgoing link, i.e., $(ac)$ can be active in $\mathbf p.$ In other words, $p_{ac} = 0, \forall (ac)\not = (ab), (ac)\in \mathcal E.$ Since $(ac)$ is not in the link decision set, $$\hat{p}_{ac} = p_{ac} = 0.$$ Hence,
      \begin{align*}
        \Pr((ab) \in \mathcal D_l | \mathbf p(t) = \mathbf p) &= \frac{1}{d_a} \Pr(a \in \mathcal D | \mathbf p(t) = \mathbf p)\\
        &= \frac{1}{d_a} \Pr(a \in \mathcal D | \mathbf p(t) = \mathbf{\hat p})\\
        &= \Pr((ab) \in \mathcal D_l | \mathbf p(t) = \mathbf p),
      \end{align*}
      Notice that the equalities hold no matter $p_{ab}>0$ or not.
\end{enumerate}
Another observation is that, given a node $a$ is in $\mathcal D,$ the event $(ab)\in \mathcal D_l$ is independent of the other links $(cd)$ with different transmitter. Thus, we have
\begin{align*}
    &\Pr(\mathcal D_l | \mathbf p(t) = \mathbf p)\\
     =& \Pr(\mathcal D_l, \mathcal D| \mathbf p(t) = \mathbf p)\\
    =& \Pr(\mathcal D_l | \mathcal D, \mathbf p(t) = \mathbf p) \Pr(\mathcal D|\mathbf p(t) = \mathbf p)\\
    =&\prod_{(ab)\in \mathcal D_l} \Pr((ab)\in \mathcal D_l | \mathcal D, \mathbf p(t) = \mathbf p) \Pr(\mathcal D|\mathbf p(t) = \mathbf p)\\
    =& \prod_{(ab)\in \mathcal D_l} \Pr((ab)\in \mathcal D_l | \mathcal D, \mathbf p(t) = \mathbf{\hat p})\Pr(\mathcal D|\mathbf p(t) = \mathbf{\hat p})\\
    =& \Pr(\mathcal D_l | \mathbf p(t) = \mathbf{\hat p})
\end{align*}

Next, we show that
\begin{equation}\label{eq: weakreversible}
    P(\mathbf{\hat{p}}, \mathbf p | \mathcal D_l) >0
\end{equation}
Recall that in the power control algorithm, if $(ab)\in \mathcal D_l,$ then link $(ab)$ randomly generate a power level according to the distribution in equalities (\ref{eq: algprob}) and (\ref{eq: algdensity}), which is always positive for any candidate power level $p\in [0, p_0^{\rm max}].$
Thus, if $(ab)\in \mathcal D_l,$ and link $(ab)$ can change its power from $p_{ab}$ to $\hat{p}_{ab}$ with positive probability, it then can change its power from $\hat{p}_{ab}$ to $p_{ab}.$ This holds for each link in the link decision set. As a result, we have equality (\ref{eq: weakreversible}). The statements in the proposition are then proved.
\end{proof}

\section{Proof of Lemma \ref{lem: conditionalreversible}}\label{sec: conditionalReversible}
\begin{proof}
By the definition of $\pi(\mathbf p):$
$$\pi(\mathbf p) = \frac{1}{Z(K_{t})} e^{(\tilde V(\mathbf p) - \epsilon \sum_{(ab)\in \mathcal E} p_{ab})/K_{t}}$$
we have
$$\frac{\pi(\mathbf p)}{\pi(\mathbf{\hat{p}})} = e^{[(\tilde V(\mathbf p) - \epsilon \sum_{(ab)\in \mathcal E} p_{ab}) - (\tilde V(\mathbf{\hat{p}}) - \epsilon \sum_{(ab)\in \mathcal E} \hat{p}_{ab})]/K_{t}}.$$

We now show how the weight $\tilde V(\mathbf p)$ is divided into local weight according to the link decision set $\mathcal D_l.$ With a little abuse of notation, we define
$$\mathcal N^1(\mathcal D_l) \triangleq \cup_{(ab)\in \mathcal D_l}\{(xy): y\in \mathcal N^1(a)\}.$$

For each link $(xy),$ one of the following two situations must be true.
\begin{enumerate}
  \item $y\in \mathcal N^1(a)$ for some node $a,$ such that $(ab)\in \mathcal D_l;$ while $y\notin \mathcal N^1(c)$ for any other node $c\not=a, (cd)\in \mathcal D_l.$
  \item $y\notin \mathcal N^1(\mathcal D_l).$
\end{enumerate}
The first situation is true, because otherwise $y\in \mathcal N^1(a) \cap \mathcal N^1(c),$ which contradicts Lemma \ref{lem: timeHomoMC}.

Therefore, we have
\begin{align*}
  \tilde V(\mathbf p) =& \sum_{(ab)\in \mathcal E} \tilde r_{ab}(\mathbf p) q_{ab} \\
   =& \sum_{(xy)\in \mathcal N^1(\mathcal D_l)} \tilde r_{xy}(\mathbf p) q_{xy} + \sum_{(xy)\notin \mathcal N^1(\mathcal D_l)} \tilde r_{xy}(\mathbf p) q_{xy}\\
  =& \sum_{(ab)\in \mathcal D_l} \sum_{(xy): y\in \mathcal N^1(a)}\tilde r_{xy}(\mathbf p) q_{xy} + \sum_{(xy)\notin \mathcal N^1(\mathcal D_l)} \tilde r_{xy}(\mathbf p) q_{xy}\\
  =& \sum_{(ab)\in \mathcal D_l} \tilde{V}_{ab}(\mathbf p) +  \sum_{(xy)\notin \mathcal N^1(\mathcal D_l)} \tilde r_{xy}(\mathbf p) q_{xy}
\end{align*}

Similarly,
\begin{equation*}
  \tilde V(\mathbf{\hat{p}}) = \sum_{(ab)\in \mathcal D_l} \tilde{V}_{ab}(\mathbf{\hat{p}}) +  \sum_{(xy)\notin \mathcal N^1(\mathcal D_l)} \tilde r_{xy}(\mathbf{\hat{p}}) q_{xy}
\end{equation*}

Note that if $(xy)\notin \mathcal N^1(\mathcal D_l),$ then the virtual rate of $(xy)$ remains the same when the power is changed from $\mathbf p$ to $\mathbf{\hat{p}},$ because none of $y'$s neighbor changes its transmit power.
Therefore, we obtain
$$\tilde V(\mathbf p) -   \tilde V(\mathbf{\hat{p}}) = \sum_{(ab)\in \mathcal D_l} \tilde{V}_{ab}(\mathbf p) - \sum_{(ab)\in \mathcal D_l} \tilde{V}_{ab}(\mathbf{\hat{p}}).$$

Notice that only the links in $\mathcal D_l$ changes their power level, it is obvious that
$$\sum_{(ab)\in \mathcal E}p_{ab} - \sum_{(ab)\in \mathcal E}\hat{p}_{ab} = \sum_{(ab)\in \mathcal D_l}p_{ab} - \sum_{(ab)\in \mathcal D_l}\hat{p}_{ab}.$$

Hence,
\begin{align}\label{eq: reversible_1}
 \frac{\pi(\mathbf p)}{\pi(\mathbf{\hat{p}})} &= e^{[(\tilde V(\mathbf p) - \epsilon \sum_{(ab)\in \mathcal E} p_{ab}) - (\tilde V(\mathbf{\hat{p}}) - \epsilon \sum_{(ab)\in \mathcal E} \hat{p}_{ab})]/K_{t}} \nonumber \\
 &= e^{\sum_{(ab)\in \mathcal D_l}[(V_{ab}(\mathbf p) - \epsilon p_{ab}) - (V_{ab}(\mathbf{\hat{p}}) - \epsilon \hat{p}_{ab})]/K_{t}}\\\nonumber
  &= \prod_{(ab)\in \mathcal D_l} \frac{e^{[V_{ab}(\mathbf p) - \epsilon p_{ab}]/K_{t}} }{e^{[V_{ab}(\mathbf{\hat{p}}) - \epsilon \hat{p}_{ab}]/K_{t}}}.
\end{align}

Since each link $(ab)\in \mathcal D_l$ in the decision set updates its transmit power level independent of the other links in the decision set, given the transmit power of the links $\{(xy): x\in \mathcal N^1(a) \cup \mathcal N^2(a)\},$ we have
\begin{align*}
  &P(\mathbf p, \mathbf{\hat{p}} | \mathcal D_l)\\
   =& \prod_{(ab)\in \mathcal D_l} g(\hat{p}_{ab}| p_{xy}:x\in \mathcal N^1(a) \cup \mathcal N^2(a), (ab)\in \mathcal D_l) \\
   =& \prod_{(ab)\in \mathcal D_l} \frac{1}{Z_{ab}(K_{t})} e^{[V_{ab}(p_{ab}) - \epsilon p_{ab}]/K_{t}}
\end{align*}
Notice that $Z_{ab}(K_{t})$ depends only on the power level of the links $\{(xy): x\in \mathcal N^1(a) \cup \mathcal N^2(a)\},$ and the power level of these links are the same in $\mathbf p$ as in $\mathbf{\hat{p}}.$ We have
\begin{align*}
\frac{P(\mathbf p, \mathbf{\hat{p}} | \mathcal D_l)}{P(\mathbf{\hat{p}}, \mathbf p | \mathcal D_l)}=& \prod_{(ab)\in \mathcal D_l} \frac{e^{[V_{ab}(\mathbf{\hat{p}}) - \epsilon \hat{p}_{ab}]/K_{t}} }{e^{[V_{ab}(\mathbf p) - \epsilon p_{ab}]/K_{t}}}\\
=& \frac{\pi(\mathbf{\hat{p}})}{\pi(\mathbf p)},
\end{align*}
which yields the result in Lemma \ref{lem: conditionalreversible}.
\end{proof}

\section{Proof of Theorem \ref{thm: Optimal}}\label{sec: proofOptimal}
In Gibbs sampler, the states of the links are updated in sequential and in a deterministic order. This deterministic updating scheme makes sure that a full sweep is finished in $n$ time slots, which establish a lower bound on the transition density between any two power configurations. We then show that under our stochastic updating scheme, a similar lower bound on the transition density can be obtained. Please notice that all the power mentioned in this proof are virtual power, we omit the tilde here for simplicity of notation.

First, let $d_{a,{\rm tx}} = |\{(xy): x\in \mathcal N^1(a) \cup \mathcal N^2(a)\}|,$ which is the number of transmitters who contend with $a$ for being in the decision set. Recall that $d_a$ denotes the outgoing degree of node $a,$ i.e., $d_a = |\{b: (ab)\in {\mathcal E}\}|.$ Further, let
$d_{\cal G} =\max\left\{ \max_{(ab)\in \mathcal E} d_{ab}, \max_{a\in \mathcal V} d_{a,{\rm tx}}\right\}.$

For the simplicity of notations, define $$P(t_2, \mathbf p_2 | t_1, \mathbf p_1) \triangleq P(\mathbf p(t_2) = \mathbf p_2 | \mathbf p(t_1)=\mathbf p_1),$$ and
$$P(t_2, \mathbf p_2 | t_1, \pi) \triangleq \int_{\mathcal P} P(t_2, \mathbf p_2 | t_1, \mathbf p) \pi(\mathbf p) d{\mathbf p}.$$

Let $||\mu - \nu||$ denote the $L^1$ distance between two distributions on $\mathcal P:$
$||\mu - \nu|| = \int_{\mathcal P} |\mu(\mathbf p) - \nu(\mathbf p)| d \mathbf p.$
Note that if $||\mu_n - \mu|| \rightarrow 0,$ then $\mu_n \rightarrow \mu$ as $n\rightarrow \infty,$ except perhaps on a subset of $\mathcal P$ with Lebesque measure 0.

\begin{lemma}\label{lem: decisionSet}
Define $\tau =  \lceil \frac{2 (d_{\cal G}+1) \log n}{c_1}\rceil,$ where $c_1=0.18$ is a constant. Let $F(t)$ denote the event that each transmitter is selected in the decision set at least twice during $[t+1, t+ 2\tau],$ then we have
$$\Pr(F(t)) \geq 1-2n^{-1}, t = 1, 2, 3,\ldots.$$
\end{lemma}
\begin{proof}
First, we show that with stochastic updating, all the transmitters are selected in the decision set at least twice during $[t+1, t+2\tau]$ with high probability, for any $t = 1, 2, 3,\ldots$

Let us consider an arbitrary transmitter $a.$ Notice that in the decision set generation stage, the decision set is generated independent of the state of the transmitter, and is independent between different time slots.
\begin{align}
   \Pr(a\in \mathcal D) &= \Pr(T_a = 0, T_x>0, \forall x\in \mathcal N^1(a)\cup \mathcal N^2(a))\nonumber\\
    &+ \Pr(T_a = 1, T_x>1, \forall x\in \mathcal N^1(a)\cup \mathcal N^2(a)) +\ldots \nonumber \\
   &+ \Pr(T_a = d_{\cal G} - 1, T_x=d_{\cal G}, \forall x\in \mathcal N^1(a)\cup \mathcal N^2(a)) \nonumber \\
   &= \frac{1}{d_{\cal G}} \left(1-\frac{1}{d_{\cal G}}\right)^{d_{ab}} + \frac{1}{d_{\cal G}} \left(1-\frac{2}{d_{\cal G}}\right)^{d_{ab}}+ \cdots\nonumber \\
   &+ \frac{1}{d_{\cal G}} \left(\frac{1}{d_{\cal G}}\right)^{d_{ab}} \nonumber \\
   &= \sum_{k=0}^{d_{\cal G} - 1} \frac{1}{d_{\cal G}} \left(1-\frac{k}{d_{\cal G}}\right)^{d_{ab}} \nonumber \\
   &= d_{\cal G} ^{-d_{ab}-1} \sum_{k=0}^{d_{\cal G} - 1} k^{d_{ab}}\nonumber \\
   &\geq d_{\cal G} ^{-d_{ab}-1} \int_{x=0}^{d_{\cal G} - 1} x^{d_{ab}} dx \nonumber \\
   &= d_{\cal G} ^{-d_{ab}-1} \frac{1}{d_{ab} + 1} (d_{\cal G}-1)^{d_{ab} + 1} \nonumber \\
   &= \left(1 - \frac{1}{d_{\cal G}}\right) ^ {d_{ab} + 1} \frac{1}{d_{ab} + 1} \nonumber \\
   & \geq  \left(1 - \frac{1}{d_{\cal G}}\right) ^ {d_{\cal G} + 1} \frac{1}{d_{ab} + 1} \nonumber \\
   & \geq  e^{-1}\left(1 - \frac{1}{d_{\cal G}}\right) \frac{1}{d_{ab} + 1} \nonumber \\
   & \geq  c_1 \frac{1}{d_{\cal G}+1},
\end{align}
where $c_1 = 0.18 < e^{-1} \frac{1}{2} \leq e^{-1}\left(1 - \frac{1}{d_{\cal G}}\right)$ is a constant.

Let $S_{a}(t)$ denote the event that transmitter $a$ is selected at least once in the decision set during time interval $[t+1, t+ \tau].$ Since the decision set is generated at each time slot independently of the previous ones, we have
\begin{align*}
  \Pr(S_{a}(t)) &= 1 - \prod_{k=t+1}^{t+\tau} (1 - \Pr(a\in \mathcal D))  \\
   &= 1 - (1 - \Pr (a\in \mathcal D))^{\tau}\\
   & \geq  1- \left(1 - \frac{c_1}{(d_{\cal G}+1)}\right)^  {\frac{2 (d_{\cal G}+1) \log n}{c_1}}\\
  &\geq  1- e^{-\frac{2 c_1 (d_{\cal G}+1) \log n }{(d_{\cal G}+1) c_1}}\\
  &= 1 - e^{-2\log n}\\
  &= 1 - n^{-2}.
\end{align*}
Then, let $H_{a}(t)$ denote the event that $a$ is selected at least twice in the decision set during $[t+1, t+2\tau].$ Obviously, if $a$ is in the decision set at least once during $[t+1, t+\tau],$ and during $[t+\tau+1, t+2\tau],$ then $H_{a}(t)$ is true, and we have
$$\Pr(H_{a}(t)) > \Pr(S_{a}(t)) \Pr(S_{a}(t+\tau)) = (1-n^{-2})^2 > 1-2n^{-2}.$$

Let $F(t)$ denote the event that every transmitter is selected in the decision set at least twice during $[t+1, t+ 2\tau],$ i.e.,$F(t) = \bigcap_{a\in \mathcal V}H_{a}(t).$ By using the union bound, we then have
\begin{align*}
  \Pr(F(t)) &= \Pr\left(\bigcap_{a\in \mathcal V} H_a(t)\right) \\
   &= 1 - \Pr\left(\bigcup_{a\in \mathcal V}\bar H_a(t)\right) \\
   &\geq 1 - \sum_{a\in \mathcal V} \Pr(\bar{H}_a(t)) \\
   & \geq  1 - n \sup_{a\in \mathcal V} \Pr(\bar{H}_a(t))\\
   & =  1 - n \left(1 - \inf_{a\in \mathcal V}\Pr(H_a(t))\right)\\
   & \geq  1 - n \cdot 2n^{-2}\\
   &\geq 1 - 2n^{-1},
\end{align*}
where the second inequality holds because the number of transmitters is no more thans the number of links in the network, i.e.,$|\mathcal V|\leq n.$
Thus, we have shown that during a time interval of length $2\tau,$ where $\tau = O(d_{\cal G} \log n),$ with probability greater than $1-2n^{-1},$ each transmitter is selected in the decision set at least twice.
\end{proof}
Next, we will show that given $F(t),$ $P(t+2\tau, \mathbf{\check{p}} | t, \mathbf{\hat{p}})>\delta(t)$ for any $\mathbf{\hat{p}}, \mathbf{\check{p}} \in \mathcal P.$ \begin{lemma} \label{lem: stochasticupdate}
$$\inf_{\mathbf{\check{p}}, \mathbf{\hat{p}} \in \mathcal P} P(t+2\tau, \mathbf{\check{p}}| t, \mathbf{\hat{p}})\geq \delta(t)$$ for any $t=1,2,\ldots,$
where $\delta(t) = c_2 e^{-2n\Delta/K_{t+2\tau}},$ and $c_2 = (1 - 2n^{-1}) (p_{a}^{\rm max}d_{\mathcal G})^{-2n}.$
\end{lemma}
\begin{proof}
Notice that in the power control algorithm, not every link whose transmitter is in the decision set can update its power. We define $\mathcal D_l$ to be the set of links who can update their transmit powers, called the link decision set. The intuition of this Lemma is that, given $F(t),$ each link has the chance to update its power, and the transition density between any two power levels of the link is lower bounded. First, if $a$ has only one outgoing link $(ab),$ then $(ab)$ must have the chance to update its power when $a$ is in the decision set. On the other hand, if $a$ has more than one outgoing link, and assume $p_{ac}>0$ at time $t,$ $\hat{p}(ab)>0$ at time $t+2\tau.$ Since it is in the decision set twice, it can turn off $(ac)$ at the first time it is in $\mathcal D,$ then it can change the power of link $(ab)$ at the second time that it is in $\mathcal D.$ Let us then derive the lower bound on $P(t+2\tau, \mathbf{\check{p}} | t, \mathbf{\hat{p}}).$

For a transmitter $a$ that has only one outgoing link $(ab),$ assume that the last time $a\in \mathcal D$ is at time slot $t_a.$ We then know $(ab)\in \mathcal D_l$ at time $t_a,$ and
\begin{align*}
    &P(p_{ab}(t+2\tau) = \check{p}_{ab} | \mathbf p(t_a - 1), (ab)\in \mathcal D_l)\\
    =& P(p_{ab}(t_a) = \check{p}_{ab} | \mathbf p(t_a - 1), (ab)\in \mathcal D_l)\\
    =&\frac{1}{Z_{ab}(K_t)}e^{[\tilde V_{ab}(p_{ab}) - \epsilon(p_{ab})]/K_t}\\
    \geq& \frac{e^{[\tilde V_{ab}(p_{ab}) - \epsilon(p_{ab})]/K_t}}{ \int_{[0, p_a^{\rm max}]} e^{[\tilde V_{ab}(\hat{p}_{ab}) - \epsilon(\hat{p}_{ab})]/K_t} d\hat{p}_{ab}} \\
    \geq& \frac{1}{p_{a}^{\rm max}}e^{(U_*-U^*)/K_{t_a}}\\
    \geq& \frac{1}{p_{a}^{\rm max}} e^{-\Delta/K_{t+2\tau}}\\
    \geq& e^{-2\Delta/K_{t+2\tau}} (p_{a}^{\rm max}d_{\mathcal G})^{-2}.
\end{align*}
For transmitter $a$ that has more than one outgoing link, assume that the last two times $a\in \mathcal D$ is at time slots $t_{a1}$ and $t_{a2},$ respectively. Notice that only one of outgoing links of $a$ can be active at each time slot. Suppose at time $t_{a1},$ link $(ac)$ is active, and in $\mathbf{\check{p}},$ at time $t+2\tau,$ link $(ab)$ is active.
\begin{align*}
    &P(p_{ab}(t+2\tau) = \check{p}_{ab}|\mathbf p(t) = \mathbf{\hat{p}})\\
     =& P(p_{ab}(t_{a2}) = \check{p}_{ab}| \mathbf p(t_{a2} - 1), (ab)\in \mathcal D_l) \\ &P(p_{ac}(t_{a1})=0 | \mathbf p(t_{a1} - 1), (ac)\in \mathcal D_l)\\
    \geq& \frac{1}{p_{a}^{\rm max}}e^{-\Delta/K_{t_{a1}}} d_{a}^{-1} \frac{1}{p_{a}^{\rm max}} e^{-\Delta/K_{t_{a2}}} d_{a}^{-1}\\
    \geq& e^{-2\Delta/K_{t+2\tau}} (p_{a}^{\rm max}d_{\mathcal G})^{-2}.
\end{align*}
Clearly, if none of the outgoing link of $a$ is active at time $t_{a1},$ then
\begin{align*}
    &P(p_{ab}(t+2\tau) = \check{p}_{ab}|\mathbf p(t) = \mathbf{\hat{p}})\\
     =& P(p_{ab}(t_{a2}) = \check{p}_{ab}| \mathbf p(t_{a2} - 1), (ab)\in \mathcal D_l) \\
    \geq& \frac{1}{p_{a}^{\rm max}} e^{-\Delta/K_{t_{a1}}} d_{a}^{-1}\\
    \geq& e^{-2\Delta/K_{t+2\tau}}  (p_{a}^{\rm max}d_{\mathcal G})^{-2}.
\end{align*}
Thus, given each transmitter is in the decision at least twice during $[t+1, t+2\tau],$, i.e., given $F(t),$ we have
\begin{align*}
    P(t+2\tau, \mathbf{\check{p}} | t, \mathbf{\hat{p}}) \geq & \prod_{a\in \mathcal V} e^{-2\Delta/K_{t+2\tau}}  (p_{a}^{\rm max}d_{\mathcal G})^{-2}\\
    \geq & e^{-2n\Delta/K_{t+2\tau}}  (p_{a}^{\rm max}d_{\mathcal G})^{-2n}.
\end{align*}

Together with $\Pr(F(t)) > 1 - 2 n^{-1},$ we have
$$P(t+2\tau, \mathbf{\check{p}} | t, \mathbf{\hat{p}}) \geq (1 - 2 n^{-1}) e^{-2n\Delta/K_{t+2\tau}}  (p_{a}^{\rm max}d_{\mathcal G})^{-2n}.$$
Let $c_2 = (1 - 2 n^{-1}) (p_{a}^{\rm max}d_{\mathcal G})^{-2n},$ and $\delta(t) = c_2 e^{-2n\Delta/K_{t+2\tau}},$ the lemma then follows.
\end{proof}

\begin{lemma} \label{lem: densityconvergence}
Let
\begin{equation}
    K_t = \left\{\begin{array}{c c}
                  \frac{K_0}{\log(2+t)} & 0<t<N, \\
                  \frac{K_0}{\log(2+N)} & N\leq t,
                \end{array}   \right.
\end{equation}
where $N\in \mathbb N$ is a fixed integer.
Then, for every $t_0 = 0, 1, 2, \ldots,$
$$\lim_{t \rightarrow \infty} \sup_{\mathbf p_0, \mathbf{\hat{p}}_0} |P(t , \mathbf p_t | t_0, \mathbf p_0 ) - P(t , \mathbf p_t | t_0 , \mathbf{\hat{p}}_0 )| d \mathbf p_t = 0.$$
\end{lemma}

\begin{proof}
Define $T_k = t_0 + 2k\tau,$ $k=1, 2, \ldots.$
First, we consider a fixed $\mathbf p_t.$
$$P(t, \mathbf p_t | t_0, \mathbf p_0) = \smallint_{\mathcal P} P(t, \mathbf p_t | T_1, \mathbf p)  P(T_1, \mathbf p | t_0, \mathbf p_0 ) d\mathbf p.$$
Let $\mathcal P_L$ be the set of power configurations that minimize $P(t, \mathbf p_t | T_1, \mathbf p),$ i.e.,
$$\mathcal P_L = \arg\min_{\mathbf p \in \mathcal P} P(t, \mathbf p_t | T_1, \mathbf p),$$
and $\mathbf p_*$ be an arbitrary element in $\mathcal P_L.$
Further, for given $\varepsilon>0,$ we define a small neighborhood around the power configurations in $\mathcal P_L:$ $$\mathcal P_{L, \varepsilon} \triangleq \{\mathbf p: |P(t, \mathbf p_t | T_1, \mathbf p) - P(t, \mathbf p_t | T_1, \mathbf p_*)| < \varepsilon\}.$$
Similarly, we define
$$\mathcal P_U = \arg\max_{\mathbf p \in \mathcal P} P(t, \mathbf p_t | T_1, \mathbf p),$$
and $\mathbf p^*$ be an arbitrary element in $\mathcal P_U.$ Also, we define
$$\mathcal P_{U, \varepsilon} \triangleq \{\mathbf p: |P(t, \mathbf p_t | T_1, \mathbf p) - P(t, \mathbf p_t | T_1, \mathbf p^*)| < \varepsilon\}.$$
Next, we derive an upper bound and an lower bound on $P(t, \mathbf p_t | t_0, \mathbf p_0)$ for different $\mathbf p_0.$
\begin{align*}
    &P(t, \mathbf p_t | t_0, \mathbf p_0)\\
    =&\smallint_{\mathcal P} P(t, \mathbf p_t | T_1, \mathbf p) P(T_1, \mathbf p | t_0, \mathbf p_0 ) d\mathbf p\\
    =&\smallint_{\mathcal P_{U,\varepsilon}} P(t, \mathbf p_t | T_1, \mathbf p) P(T_1, \mathbf p | t_0, \mathbf p_0 ) d\mathbf p \\
    &+\smallint_{\mathcal P_{L,\varepsilon}} P(t, \mathbf p_t | T_1, \mathbf p) P(T_1, \mathbf p | t_0, \mathbf p_0 ) d\mathbf p\\
    &+\smallint_{\mathcal P\setminus (\mathcal P_{U,\varepsilon}\cup \mathcal P_{L,\varepsilon})} P(t, \mathbf p_t | T_1, \mathbf p) P(T_1, \mathbf p | t_0, \mathbf p_0 ) d\mathbf p
\end{align*}
Since $\smallint_{\mathcal P}P(T_1, \mathbf p | t_0, \mathbf p_0 ) d\mathbf p = 1,$ and by Lemma \ref{lem: stochasticupdate}, we have $$P(T_1, \mathbf p | t_0, \mathbf p_0 )\geq \delta(t_0).$$ Hence, the maximum value of $P(t, \mathbf p_t | t_0, \mathbf p_0)$ is obtained by letting $P(T_1, \mathbf p | t_0, \mathbf p_0)=\delta(t_0)$ for all $\mathbf p \not\in \mathcal P_{U,\varepsilon}.$ Thus,
$$\smallint_{\mathcal P_{U,\varepsilon}}  P(T_1, \mathbf p | t_0, \mathbf p_0 ) \leq 1 - \sigma(\mathcal P \setminus \mathcal P_{U,\varepsilon})\delta(t_0),$$
where $\sigma$ is used to denoted the Lebesque measure of a set, in order to distinguish with the modulation function.
Hence, we have,
\begin{align*}
    &P(t, \mathbf p_t | t_0, \mathbf p_0)\\
    \leq&P(t, \mathbf p_t | T_1, \mathbf p^*)\smallint_{\mathcal P_{U,\varepsilon}}  P(T_1, \mathbf p | t_0, \mathbf p_0 ) d\mathbf p \\
    &+\smallint_{\mathcal P_{L,\varepsilon}} P(t, \mathbf p_t | T_1, \mathbf p) P(T_1, \mathbf p | t_0, \mathbf p_0 ) d\mathbf p\\
    &+\smallint_{\mathcal P\setminus (\mathcal P_{U,\varepsilon}\cup \mathcal P_{L,\varepsilon})} P(t, \mathbf p_t | T_1, \mathbf p) P(T_1, \mathbf p | t_0, \mathbf p_0 ) d\mathbf p\\
    \leq& P(t, \mathbf p_t | T_1, \mathbf p^*)[1 - \sigma(\mathcal P\setminus \mathcal P_{U,\varepsilon}) \delta(t_0)]\\
    +& (P(t, \mathbf p_t | T_1, \mathbf p_*) + \varepsilon) \sigma(\mathcal P_{L,\varepsilon}) \delta(t_0)\\
    +& \delta(t_0) \smallint_{\mathcal P\setminus (\mathcal P_{U,\varepsilon}\cup \mathcal P_{L,\varepsilon})} P(t, \mathbf p_t | T_1, \mathbf p) d\mathbf p,
\end{align*}
Similarly, we can obtain a lower bound on $P(t, \mathbf p_t | t_0, \mathbf p_0)$ as well:
\begin{align*}
    P(t, \mathbf p_t | t_0, \mathbf p_0)\geq& P(t, \mathbf p_t | T_1, \mathbf p_*)[1 - \sigma(\mathcal P\setminus \mathcal P_{L,\varepsilon}) \delta(t_0)]\\
    &+ (P(t, \mathbf p_t | T_1, \mathbf p^*) - \varepsilon) \sigma(\mathcal P_{U,\varepsilon}) \delta(t_0)\\
    &+ \delta(t_0) \smallint_{\mathcal P\setminus (\mathcal P_{U,\varepsilon}\cup \mathcal P_{L,\varepsilon})} P(t, \mathbf p_t | T_1, \mathbf p) d\mathbf p,
\end{align*}
Hence, by taking the difference between the upper bound and lower bounded derived above, we have that for any $\mathbf p_t, \mathbf p_0,$ and $\mathbf{\hat{p}}_0,$
\begin{align*}
    &|P(t, \mathbf p_t | t_0, \mathbf p_0) - P(t, \mathbf p_t | t_0, \mathbf{\hat{p}}_0)|\\
    \leq&\left|P(t, \mathbf p_t | T_1, \mathbf p^*)-P(t, \mathbf p_t | T_1, \mathbf p_*)\right|[1 - \sigma(\mathcal P) \delta(t_0)]\\
    &+ \varepsilon \delta(t_0) (\sigma(\mathcal P_{U,\varepsilon}) + \sigma(\mathcal P_{L, \varepsilon}))\\
    \leq&(P(t, \mathbf p_t | T_1, \mathbf p^*)-P(t, \mathbf p_t | T_1, \mathbf p_*))[1-\sigma(\mathcal P) \delta(t_0)]\\
    =& \sup_{\mathbf p, \mathbf{\hat{p}}} |P(t, \mathbf p_t | T_1, \mathbf p) - P(t, \mathbf p_t | T_1, \mathbf{\hat{p}})|[1-\sigma(\mathcal P) \delta(t_0)],
\end{align*}
where the last inequality holds by the definition of $\mathbf p_*$ and $\mathbf p^*,$ and also because $\varepsilon$ is arbitrarily small.
Since the inequality holds for any $\mathbf p_0, \mathbf{\hat{p}}_0,$ we have
\begin{align*}
    &\sup_{\mathbf p_0, \mathbf{\hat{p}}_0}|P(t, \mathbf p_t | t_0, \mathbf p_0) - P(t, \mathbf p_t | t_0, \mathbf{\hat{p}}_0)| \\
    \leq&\sup_{\mathbf p, \mathbf{\hat{p}}} |P(t, \mathbf p_t | T_1, \mathbf p) - P(t, \mathbf p_t | T_1, \mathbf{\hat{p}})|[1-\sigma(\mathcal P) \delta(t_0)].
\end{align*}

Let $k = \lfloor\frac{t-t_0}{2\tau} \rfloor,$ such that $T_k\leq t$ is the nearest time point to $t$ which is of the form $t_0 + 2i\tau.$ Notice that $K_t \geq \frac{2n\Delta}{\log (2+t)},$ and $\delta(t)= c_2 e^{-2n\Delta/K_{t+2\tau}},$
we have
$\delta(t_0+2i\tau) \geq\frac{c_2}{2+t_0+2(i+1)\tau}.$
Hence,
\begin{align*}
    &\sup_{\mathbf p_0, \mathbf{\hat{p}}_0} |P(t, \mathbf p_t | t_0, \mathbf p_0) - P(t, \mathbf p_t | t_0, \mathbf{\hat{p}}_0)|\\
    \leq& \sup_{\mathbf p, \mathbf{\hat{p}}}|P(t, \mathbf p_t | T_k, \mathbf p) - P(t, \mathbf p_t | T_k, \mathbf{\hat{p}})| \\
    &\prod_{i=0}^{k-1}[1-\sigma(\mathcal P)\delta(t_0+2i\tau)]\\
    \leq&\sup_{\mathbf p, \mathbf{\hat{p}}}|P(t, \mathbf p_t | T_k, \mathbf p) - P(t, \mathbf p_t | T_k, \mathbf{\hat{p}})| \\
    &\prod_{i=0}^{k-1}\left[1-\frac{c_2\sigma(\mathcal P)}{2+t_0+2(i+1)\tau}\right]
\end{align*}
By Lemma \ref{lem: transition}, it can be easily verified that $|P(t, \mathbf p_t | T_k, \mathbf p) - P(t, \mathbf p_t | T_k, \mathbf{\hat{p}})|$
is bounded for each $\mathbf p_t, \mathbf p, \mathbf{\hat{p}},$ and $t.$
Further, since $\lim_{k\rightarrow \infty}\sum_{i=0}^{k-1}(t_0+2(i+1)\tau+2) = \infty,$
we have
$$\lim_{k\rightarrow \infty} \prod_{i=0}^{k-1}\left[1-\sigma(\mathcal P)\frac{c_2}{2+t_0+2(i+1)\tau}\right] = 0.$$
Note that $k\rightarrow \infty$ as $t\rightarrow \infty,$ we then obtain
$$\lim_{t\rightarrow \infty}\sup_{\mathbf p_t, \mathbf p_0, \mathbf p_0} |P(t, \mathbf p_t | t_0, \mathbf p_0) - P(t, \mathbf p_t | t_0, \mathbf {\hat{p}}_0)| =0,$$
which concludes the lemma.
\end{proof}

\begin{lemma} \label{lem: stationaryconvergence}
For any given $\varepsilon>0,$ there exists $N\in \mathbb N$ such that for any $t\geq t_0\geq N, t^*\geq N,$ we have
$$||P(t, \cdot |  t_0, \pi_{K_{t^*}}) - \pi_{K_{t^*}}|| < \varepsilon.$$
\end{lemma}
\begin{proof}

First, we claim that for any $t>t_0\geq 0,$ and starting with any distribution $\mu,$
$$||P(t,\cdot| t_0, \mu) - \pi_{K_t}|| \leq ||P(t-1, \cdot| t_0, \mu) - \pi_{K_{t}}||.$$
\begin{align*}
        &||P(t,\cdot| t_0, \mu) - \pi_{K_t}||\\
    =&\smallint_{\mathcal P} |P(t, \mathbf{\hat{p}} | t_0, \mu) - \pi_{K_t}(\mathbf{\hat{p}})| d\mathbf{\hat{p}}\\
    =& \smallint_{\mathcal P}|\smallint_{\mathcal P_{\mathcal D_l}} P(t, (\mathbf{\hat{p}}_{\mathcal D_l} , \mathbf{\hat{p}}_{-\mathcal D_l})| t-1, (\mathbf{\check{p}}_{\mathcal D_l}, \mathbf{\hat{p}}_{-\mathcal D_l})\\
     & P(t-1, (\mathbf{\check{p}}_{\mathcal D_l}, \mathbf{\hat{p}}_{-\mathcal D_l}) | t_0, \mu) - \\
     & \pi_{K_t}((\mathbf{\check{p}}_{\mathcal D_l}, \mathbf{\hat{p}}_{-\mathcal D_l}))P(t, (\mathbf{\hat{p}}_{\mathcal D_l} , \mathbf{\hat{p}}_{-\mathcal D_l}) | t-1, (\mathbf{\check{p}}_{\mathcal D_l}, \mathbf{\hat{p}}_{-\mathcal D_l})) d\mathbf{\check{p}}_{\mathcal D_l}|d\mathbf{\hat{p}}\\
    \leq& \smallint_{\mathcal P}\smallint_{\mathcal P_{\mathcal D_l}}P(t, (\mathbf{\hat{p}}_{\mathcal D_l} , \mathbf{\hat{p}}_{-\mathcal D_l})| t-1, (\mathbf{\check{p}}_{\mathcal D_l}, \mathbf{\hat{p}}_{-\mathcal D_l}))\\
    &|P(t-1, (\mathbf{\check{p}}_{\mathcal D_l}, \mathbf{\hat{p}}_{-\mathcal D_l}) | t_0, \mu) - \pi_{K_t}((\mathbf{\check{p}}_{\mathcal D_l}, \mathbf{\hat{p}}_{-\mathcal D_l}))| d\mathbf{\check{p}}_{\mathcal D_l} d\mathbf{\hat{p}}\\
    =& \smallint_{\mathcal P_{-\mathcal D_l}}\smallint_{\mathcal P_{\mathcal D_l}} |P(t-1, (\mathbf{\check{p}}_{\mathcal D_l}, \mathbf{\hat{p}}_{-\mathcal D_l}) | t_0, \mu) \\
    & - \pi_{K_t}((\mathbf{\check{p}}_{\mathcal D_l}, \mathbf{\hat{p}}_{-\mathcal D_l}))|d\mathbf{\check{p}}_{\mathcal D_l}d\mathbf{\hat{p}}_{-\mathcal D_l}\\
    &\smallint_{\mathcal P_{\mathcal D_l}}P(t, (\mathbf{\hat{p}}_{\mathcal D_l} , \mathbf{\hat{p}}_{-\mathcal D_l})| t-1, (\mathbf{\check{p}}_{\mathcal D_l}, \mathbf{\hat{p}}_{-\mathcal D_l})) d\mathbf{\hat{p}}_{\mathcal D_l}\\
    =& \smallint_{\mathcal P_{-\mathcal D_l}}\smallint_{\mathcal P_{\mathcal D_l}} |P(t-1, (\mathbf{\check{p}}_{\mathcal D_l}, \mathbf{\hat{p}}_{-\mathcal D_l}) | t_0, \mu)\\
    & - \pi_{K_t}((\mathbf{\check{p}}_{\mathcal D_l}, \mathbf{\hat{p}}_{-\mathcal D_l}))|d\mathbf{\check{p}}_{\mathcal D_l}d\mathbf{\hat{p}}_{-\mathcal D_l}\\
    =& \smallint_{\mathcal P}|P(t-1, (\mathbf{\check{p}}_{\mathcal D_l}, \mathbf{\hat{p}}_{-\mathcal D_l}) | t_0, \mu) - \pi_{K_t}((\mathbf{\check{p}}_{\mathcal D_l}, \mathbf{\hat{p}}_{-\mathcal D_l}))|d\mathbf{p}\\
    =& ||P(t-1,\cdot| t_0, \mu) - \pi_{K_t}||,
\end{align*}
where the second equality holds by Lemma \ref{lem: Stationary}, i.e.,
\begin{align*}&\pi_{K_t}(\mathbf{\hat{p}}) =\\
& \smallint_{\mathcal P_{\mathcal D_l}}\pi_{K_t}((\mathbf{\check{p}}_{\mathcal D_l}, \mathbf{\hat{p}}_{-\mathcal D_l}))P(t, (\mathbf{\hat{p}}_{\mathcal D_l} , \mathbf{\hat{p}}_{-\mathcal D_l}) | t-1, (\mathbf{\check{p}}_{\mathcal D_l}, \mathbf{\hat{p}}_{-\mathcal D_l})) d\mathbf{\check{p}}_{\mathcal D_l}.
\end{align*}

Observe that as $t \rightarrow \infty,$ $\pi_{K_t}$ will have higher probability in each optimal power configuration. It can be shown that there exists an $N\in\mathbb{N}$ large enough, such that for $t\geq N,$ $\pi_{K_t}(\mathbf p)$ is strictly increasing in $t,$ for each $\mathbf p \in \mathcal P^*;$ while $\pi_{K_t}(\mathbf p)$ is strictly decreasing in $t,$ for each $\mathbf p\notin \mathcal P^*.$
Thus, we have
$$\sum_{t=1}^\infty ||\pi_{K_t} - \pi_{K_{t+1}}|| < \infty.$$

We then have that, starting with distribution $\pi_{K_{t^*}}$ at time $t_0,$
\begin{align*}
    &||P(t, \cdot| t_0, \pi_{K_{t^*}}) -  \pi_{K_{t^*}}|| \\
\leq& ||P(t, \cdot| t_0, \pi_{K_{t^*}}) - \pi_{K_t}|| + ||\pi_{K_{t}} - \pi_{K_{t^*}}||\\
\leq& ||P(t-1, \cdot| t_0, \pi_{K_{t^*}}) - \pi_{K_t}|| + ||\pi_{K_{t}} - \pi_{K_{t^*}}||\\
\leq& ||P(t-1, \cdot| t_0, \pi_{K_{t^*}}) - \pi_{K_{t-1}}||\\
& + ||\pi_{K_{t-1}} - \pi_{K_{t}}|| + ||\pi_{K_{t}} - \pi_{K_{t^*}}||\\
&\vdots\\
\leq& ||P(t_0, \cdot| t_0, \pi_{K_{t^*}}) - \pi_{K_{t_0}}||\\
& + \sum_{m=t_0}^{t-1} ||\pi_{K_m} - \pi_{K_{m+1}}|| + ||\pi_{K_m} - \pi_{K_{t^*}}||
\end{align*}

The first term, $||P(t_0, \cdot| t_0, \pi_{K_{t^*}}) - \pi_{K_{t_0}}||=0,$ because we assume that the process starts with distribution $\pi_{K_{t_0}} = \pi_{K_{t^*}}.$
Since $\sum_{t=1}^\infty ||\pi_{K_{t}} - \pi_{K_{t+1}}|| < \infty,$ we have that, for given $\varepsilon >0,$ there exist $N\in \mathbb{N},$ such that for $t>t_0\geq N,t^*\geq N,$
$$\sum_{m=t_0}^{t-1} ||\pi_{K_m} - \pi_{K_{m+1}}|| < \frac{\varepsilon}{2},$$
and
$$||\pi_{K_{t}} - \pi_{K_{t^*}}|| < \sum_{m=\min\{t^*,t\}}^{\max\{t^*,t\} - 1} ||\pi_{K_{m}} - \pi_{K_{m+1}}|| <\frac{\varepsilon}{2}.$$

The lemma then follows.
\end{proof}

Finally, we show how to establish Theorem \ref{thm: Optimal} by Lemma \ref{lem: densityconvergence} and Lemma \ref{lem: stationaryconvergence}.
\begin{align}\label{eq: z1}
&||P(t,\cdot| 0,\mathbf p) - \pi_{K_{t^*}}||\nonumber\\
=& \left\|\smallint_{\mathcal P} P(t,\cdot| t_0, \mathbf{\hat{p}}) P(t_0, \mathbf{\hat{p}} | 0, \mathbf p) d\mathbf{\hat{p}} - \pi_{K_{t^*}}\right\|\nonumber\\
\leq & \left\|\smallint_{\mathcal P} P(t, \cdot| t_0, \mathbf{\hat{p}})P(t_0, \mathbf{\hat{p}} | 0, \mathbf p) d\mathbf{\hat{p}} - P(t, \cdot| t_0, \pi_{K_{t^*}})\right\|\nonumber\\
+&\|P(t, \cdot| t_0, \pi_{K_{t^*}}) - \pi_{K_{t^*}} \|
\end{align}

By Lemma \ref{lem: stationaryconvergence}, for given $\varepsilon >0,$ there exists $N_1\in \mathbb N,$ such that if $t\geq t_0\geq N_1,$ $t^*\geq N_1,$
\begin{equation}\label{eq: z2}
    \|P(t, \cdot| t_0, \pi_{K_{t^*}}) - \pi_{K_{t^*}} \| < \frac{\varepsilon}{4}
\end{equation}

Consider the first term,
\begin{align*}
& \left\|\smallint_{\mathcal P} P(t, \cdot| t_0, \mathbf{\hat{p}})P(t_0, \mathbf{\hat{p}} | 0, \mathbf p) d\mathbf{\hat{p}} - P(t, \cdot| t_0, \pi_{K_{t^*}})\right\|\\
=& \left\|\smallint_{\mathcal P} P(t, \cdot| t_0, \mathbf{\hat{p}})P(t_0, \mathbf{\hat{p}} | 0, \mathbf p) d\mathbf{\hat{p}}\right.\\
 &- \left.\smallint_{\mathcal P} P(t, \cdot| t_0, \mathbf{\hat{p}}) \pi_{K_{t^*}}(\mathbf{\hat{p}}) d\mathbf{\hat{p}}\right\|\\
=& \left\|\smallint_{\mathcal P}  P(t, \cdot| t_0, \mathbf{\hat{p}}) (P(t_0, \mathbf{\hat{p}} | 0, \mathbf p) - \pi_{K_{t^*}}(\mathbf{\hat{p}}))\right\|\\
=& \smallint_{\mathcal P} \left|\smallint_{\mathcal P} P(t, \mathbf p| t_0, \mathbf{\hat{p}}) (P(t_0, \mathbf{\hat{p}} | 0, \mathbf p) - \pi_{K_{t^*}}(\mathbf{\hat{p}})) d\mathbf{\hat{p}}\right| d\mathbf p\\
=& \smallint_{\mathcal P}\left|\smallint_{\mathcal P}(P(t, \mathbf p| t_0, \mathbf{\hat{p}}) - P(t, \mathbf p| t_0, \mathbf{\check{p}}))\right.\\
 &\left.(P(t_0, \mathbf{\hat{p}} | 0, \mathbf p) - \pi_{K_{t^*}}(\mathbf{\hat{p}})) d\mathbf{\hat{p}}\right| d\mathbf p\\
\leq&  \smallint_{\mathcal P}  \smallint_{\mathcal P}\left|(P(t, \mathbf p| t_0, \mathbf{\hat{p}}) - P(t, \mathbf p| t_0, \mathbf{\check{p}}))\right|\\
 &\left|(P(t_0, \mathbf{\hat{p}} | 0, \mathbf p) - \pi_{K_{t^*}}(\mathbf{\hat{p}}))\right| d\mathbf{\hat{p}} d\mathbf p
\end{align*}

By Lemma \ref{lem: densityconvergence}, for any given $\varepsilon >0,$ there exists $N\in \mathbb N,$ such that for $t\geq N,$
$$\left|(P(t, \mathbf p| t_0, \mathbf{\hat{p}}) - P(t, \mathbf p| t_0, \mathbf{\check{p}}))\right|\leq \frac{\varepsilon}{4\prod_{a\in \mathcal V} p_a^{\rm max}}.$$
Thus,
\begin{align}\label{eq: z3}
& \left\|\smallint_{\mathcal P} P(t, \cdot| t_0, \mathbf{\hat{p}})P(t_0, \mathbf{\hat{p}} | 0, \mathbf p) d\mathbf{\hat{p}} - P(t, \cdot| t_0, \pi_{K_{t^*}})\right\| \nonumber\\
\leq&  \smallint_{\mathcal P} \smallint_{\mathcal P}\left|(P(t, \mathbf p| t_0, \mathbf{\hat{p}}) - P(t, \mathbf p| t_0, \mathbf{\check{p}}))\right|\nonumber\\
 &\left|(P(t_0, \mathbf{\hat{p}} | 0, \mathbf p) - \pi_{K_{t^*}}(\mathbf{\hat{p}}))\right| d\mathbf{\hat{p}} d\mathbf p\nonumber\\
\leq & \smallint_{\mathcal P}  \smallint_{\mathcal P} \frac{\varepsilon}{8\prod_{a\in \mathcal V} p_a^{\rm max}} \left|(P(t_0, \mathbf{\hat{p}} | 0, \mathbf p) - \pi_{K_{t^*}}(\mathbf{\hat{p}}))\right|d\mathbf{\hat{p}} d\mathbf p\nonumber\\
\leq & \smallint_{\mathcal P} \frac{\varepsilon}{4\prod_{a\in \mathcal V} p_a^{\rm max}}\nonumber\\
\leq & \frac{\varepsilon}{4}.
\end{align}

Hence, by equalities (\ref{eq: z1}),(\ref{eq: z2}) and (\ref{eq: z3}),
$$\|P(t, \cdot| 0, \mathbf p) - \pi_{K_{t^*}}\| < \frac{\varepsilon}{2}.$$

Notice that for given $\varepsilon>0,$ there exists $N_2\in \mathbb N$ such that for $t^* \geq N_2,$ we have
\begin{equation}\label{eq: z4}
\smallint_{\mathcal P_\delta^*} \pi_{K_{t^*}}(\mathbf p) d\mathbf p > 1 - \varepsilon/2.
\end{equation}
By selecting $N = \max\{N_1, N_2\},$ inequalities (\ref{eq: z2}) and (\ref{eq: z4}) hold, which imply that
\begin{align*}
    &\smallint_{\mathcal P_{\delta}^*} P(t, \mathbf{\hat{p}}| 0, \mathbf p) d\mathbf{\hat{p}}\\
\geq&\smallint_{\mathcal P_{\delta}^*} \pi_{K_{t^*}}(\mathbf{\hat{p}}) - |P(t, \mathbf{\hat{p}}| 0, \mathbf p) -  \pi_{K_{t^*}}(\mathbf{\hat{p}})| d\mathbf{\hat{p}}\\
= &  \smallint_{\mathcal P_{\delta}^*} \pi_{K_{t^*}}(\mathbf{\hat{p}}) d\mathbf{\hat{p}} - \smallint_{\mathcal P_{\delta}^*} |P(t, \mathbf{\hat{p}}| 0, \mathbf p) -  \pi_{K_{t^*}}(\mathbf{\hat{p}})| d\mathbf{\hat{p}}\\
>& 1 - \frac{\varepsilon}{2} - \frac{\varepsilon}{2}\\
>& 1 - \varepsilon.
\end{align*}

We then conclude the theorem.

\bibliographystyle{abbrv}
\bibliography{./Shan_reference}  
%
%
\end{document}